\documentclass[11pt]{article}

\usepackage[final]{acl}
\usepackage{afterpage}
\usepackage{times}
\usepackage{latexsym}
\usepackage[T1]{fontenc}

\usepackage[utf8]{inputenc}

\usepackage{microtype}
\usepackage{adjustbox}
\usepackage{inconsolata}
\usepackage{hyperref}
\usepackage{fontawesome5}

\usepackage{graphicx}
\usepackage{placeins}
\usepackage{microtype}
\usepackage{subcaption}
\usepackage{rotating}
\usepackage{tcolorbox}
\tcbuselibrary{breakable}
\usepackage{multirow} 
\usepackage{makecell}
\usepackage{booktabs} 

\usepackage{amsmath}
\usepackage{amssymb}
\usepackage{mathtools}
\usepackage{algorithmic}
\usepackage{algorithm}
\usepackage{amsthm}
\usepackage{booktabs}

\usepackage[nolist,nohyperlinks]{acronym}
\acrodef{sac}[SAC]{Self-Anchored Consensus}
\usepackage{cleveref}
\usepackage[table]{xcolor}
\usepackage[dvipsnames]{xcolor}
\crefformat{problem}{Problem~#2#1#3}
\crefformat{assumption}{Assumption~#2#1#3}

\theoremstyle{plain}
\newtheorem{theorem}{Theorem}[section]

\newtheorem{lemma}[theorem]{Lemma}

\theoremstyle{definition}
\newtheorem{definition}[theorem]{Definition}

\newtheorem{remark}[theorem]{Remark}
\newtheorem{property}[theorem]{Property}

\newcommand{\Acal}{\mathcal{A}}
\newcommand{\Bcal}{\mathcal{B}}

\newcommand{\Ecal}{\mathcal{E}}
\newcommand{\Fcal}{\mathcal{F}}
\newcommand{\Gcal}{\mathcal{G}}
\newcommand{\Hcal}{\mathcal{H}}

\newcommand{\Ncal}{\mathcal{N}}

\newcommand{\Rcal}{\mathcal{R}}
\newcommand{\Scal}{\mathcal{S}}

\newcommand{\Vcal}{\mathcal{V}}

\title{Robust Multi-Agent LLMs under Byzantine Faults}

\author{
Haejoon Lee\textsuperscript{1}$^{,\dag}$,
Vincent-Daniel Yun\textsuperscript{2}$^{,\dag}$,
Dimitra Panagou\textsuperscript{1},
Sai Praneeth Karimireddy\textsuperscript{2}
\\ \\
\textsuperscript{1}Department of Robotics, University of Michigan, Ann Arbor \\
\textsuperscript{2}Thomas Lord Department of Computer Science, University of Southern California
\\
 \small{
   \textbf{Correspondence:} Sai Praneeth Karimireddy <karimire@usc.edu>
} \\ \\
\href{https://github.com/daniel-eai/Robust-Multi-Agent-LLMs-under-Byzantine-Faults}
{\faGithub\ GitHub}
}

\begin{document}
\maketitle

\begin{abstract}
Large language model (LLM) agents increasingly collaborate over peer-to-peer networks to improve their reliability. However, these same interactions can also introduce vulnerability to unreliable or Byzantine agents that can propagate incorrect information and degrade overall system performance. To address this, we propose Self-Anchored Consensus (SAC), a fully decentralized filter-and-refine protocol in which agents iteratively exchange responses, locally evaluate and filter unreliable messages, and refine their own outputs. We present $(F{+}1)$-robustness conditions on the communication graph that ensure honest agents preserve and propagate reliable information despite Byzantine influence. Experiments across diverse open- and closed-weight LLMs on mathematical and commonsense reasoning benchmarks show that SAC effectively suppresses Byzantine influence and consistently improves performance across diverse communication topologies, whereas prior methods degrade significantly under Byzantine attacks.
 \begingroup
\renewcommand\thefootnote{}
\footnote{$^{\dag}$These authors contributed equally.}
\addtocounter{footnote}{-1}
\endgroup
\end{abstract}

\section{Introduction}
Large language models (LLMs) are increasingly deployed as communicating agents in multi-agent systems (MAS), where multiple agents exchange responses, critique one another, and converge on collective answers~\citep{guo2024large,du2023improving}. We consider a peer-to-peer setting in which LLM agents collaboratively solve tasks with objectively verifiable answers through consensus.

As these systems move toward real-world deployment, a central challenge is \emph{robustness} to faulty and adversarial agents. Individual agents may hallucinate, miscalculate, or strategically inject misleading responses. Since agents cannot locally distinguish trustworthy messages from Byzantine ones (Fig.~\ref{fig:hero}), unreliable responses can contaminate otherwise reliable agents. The key question is therefore whether Byzantine influence can be \textit{contained} so that honest agents preserve their base capability.

\begin{figure}
\centering
\includegraphics[width=0.95\linewidth]{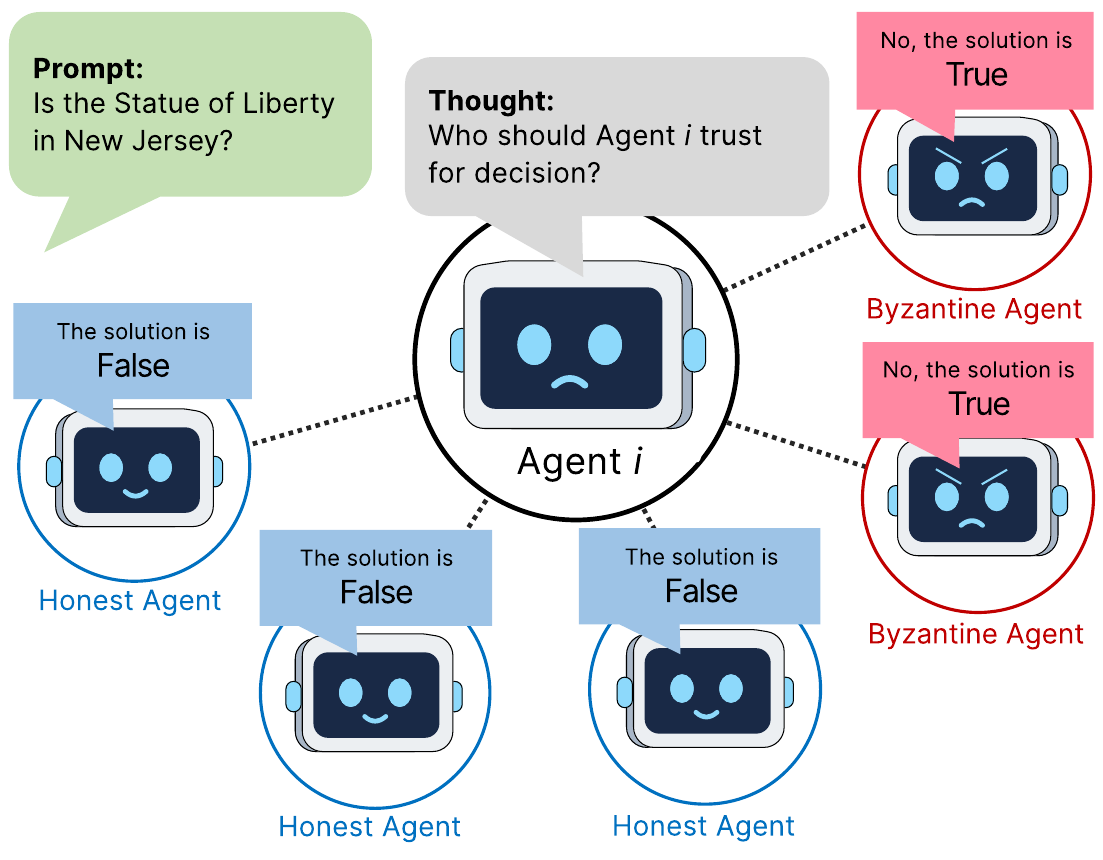}
\caption{Decision ambiguity in the presence of Byzantine agents. An agent must decide which responses to trust among indistinguishable neighbor messages.}
\label{fig:hero}
\end{figure}

This setting naturally connects to classical Byzantine-resilient consensus, where some agents behave arbitrarily while honest nodes must still reach correct agreement~\citep{lamport1982byzantine, leblanc2013resilient, su2021byzantine}. Viewing LLM-MAS through this lens raises two questions: (i) how should agents decide which neighbors to trust, and (ii) how should the communication graph be designed to limit Byzantine influence?

Recent methods such as CP-WBFT~\citep{zheng2026rethinking} aggregate neighbor responses using self-reported confidence scores. However, this implicitly assumes Byzantine agents report honest reliability signals, departing from the standard Byzantine setting where adversaries may falsify confidence values. Moreover, communication graphs are typically selected heuristically from generic topologies without principled fault-tolerance guarantees, resulting in unstable performance across different network structures~\citep{zheng2026rethinking}.

We address these limitations by drawing on classical Byzantine-resilient consensus. Specifically, we adapt Mean-Subsequence-Reduced (MSR) filtering~\citep{dolev1986reaching, kiechhafer1994reaching} and $r$-robustness~\citep{leblanc2013resilient} to LLM-MAS by (i) replacing sender-reported confidence with receiver-side evaluations computed independently by each agent, and (ii) introducing an MSR-inspired filter-and-refine mechanism for iterative response refinement. We further design communication graphs under which honest agents remain robust to Byzantine influence.

\paragraph{Contributions.}
Our contributions are twofold. First, we introduce~\acf{sac}, a decentralized MSR-inspired algorithm that enables robust LLM collaboration through receiver-side confidence evaluation and iterative filter-and-refine updates tailored to natural-language interactions. Second, we establish graph-theoretic robustness conditions under which honest agents contain Byzantine influence under~\ac{sac}. Experiments show that our method consistently preserves honest-agent capability across diverse communication topologies, whereas prior approaches collapse catastrophically under Byzantine faults.

\section{Related Works}
\label{sec:related}
\paragraph{LLM multi-agent systems.}
Multiple LLM instances coordinated as agents have been shown to outperform single-agent inference through debate~\citep{du2023improving,liang2024encouraging} and role-based frameworks~\citep{wu2023autogen,hong2024metagpt,guo2024large}. A growing body of work studies the reliability cost when such agents behave incorrectly~\citep{tian2023evil,yu2024netsafe,zhang2024cutcrap}. Our work takes this failure mode as its starting point and asks what aggregation rule and graph structure are needed to contain it.
 
\paragraph{Byzantine-robust consensus in LLM-MAS.}
Three recent methods are closest to ours. CP-WBFT~\citep{zheng2026rethinking} relies on self-reported confidence to select high-confidence responses, implicitly assuming Byzantine agents would honestly report their confidence. It also evaluates their method across the six hand-picked graphs, with wildly varying performance. Trusted MultiLLMN~\citep{luo2025weighted} uses a leader-based BFT protocol in which consecutive Byzantine leaders force expensive re-elections. DecentLLMs~\citep{jo2025byzantine} removes the leader via a geometric-mean filter over evaluator scores, which is brittle under high-variance scoring and requires a fully-connected evaluator-worker graph. Our method differs on two axes at once: confidence is computed on the \emph{receiver} side, so Byzantine agents cannot manipulate how others evaluate their output; and our communication topology is grounded on mathematically-driven robustness properties that limit Byzantine influence on general graphs.
 
\paragraph{Classical Byzantine-resilient algorithms.}
The Byzantine agreement problem~\citep{lamport1982byzantine, dolev1986reaching} has been extensively studied to ensure reliability under arbitrary node failures, inspiring a broad line of work on fault-tolerant distributed computation. In particular, the \textit{Weighted Mean-Subsequence-Reduced} (W-MSR) algorithm, together with $r$- and $(r,s)$-robustness conditions~\citep{leblanc2013resilient}, provides consensus guarantees despite Byzantine agents. These notions have been extended beyond consensus to distributed optimization~\citep{sundaram2019distributed_opt, yuan2025resilient} and distributed learning~\citep{xie2023communication, ye2024resilient}, with applications in robotics and smart grids~\citep{lee2025distributed, yuan2025resilient}. To our knowledge, this work is the first to bring MSR-type algorithms and robustness conditions to design Byzantine-resilient LLM-MAS.

\section{Problem Setup}
\label{sec:problem}

\subsection{LLM Multi-Agent System}

We consider a multi-agent system of $n$ LLM agents, indexed by $\Vcal=\{1,\dots,n\}$, collaboratively answering a query $x$ drawn from a task distribution $\mathcal{D}$. Agents exchange messages over an undirected, time-invariant communication graph $\mathcal{G}=(\Vcal,\mathcal{E})$, where an edge $(i,j)\in\mathcal{E}$ indicates bidirectional exchange between agents $i$ and $j$, and $\Ncal_i=\{j\in\Vcal \mid (i,j)\in\mathcal{E}\}$ denotes the neighbor set of agent $i$. We also define a directed graph $\Gcal'=(\Vcal,\Ecal')$, where a directed edge $(i,j)\in\Ecal'$ indicates that $j$ receives messages from $i$. A directed graph contains a rooted out-branching if there exists $p\in\Vcal$ that can reach all other nodes. Let $\mathbb{Z}_{\geq 0}$ denote the non-negative integers.

At each time step $t\in \mathbb Z_{\geq 0}$, agents communicate with their neighbors. Each agent $i$ is associated with an underlying LLM $\mathcal{M}_i$ and produces a response $r_i^{(t)}\in\mathcal{Y}$ to a query $x$, where $\mathcal{Y}$ is the space of admissible outputs (e.g., numerical answers or binary safety labels). We assume that each query $x$ admits a well-defined ground-truth answer $y^*\in \mathcal{Y}$, allowing responses to be classified as correct ($r_i^{(t)}=y^*$) or incorrect ($r_i^{(t)}\neq y^*$) at each round $t$.

\subsection{Threat Model}

We consider a setting in which a subset of agents are Byzantine, adapting the definition of~\citet{leblanc2013resilient} to LLM-MAS:

\begin{figure}
    \centering
\includegraphics[width=0.7\linewidth]{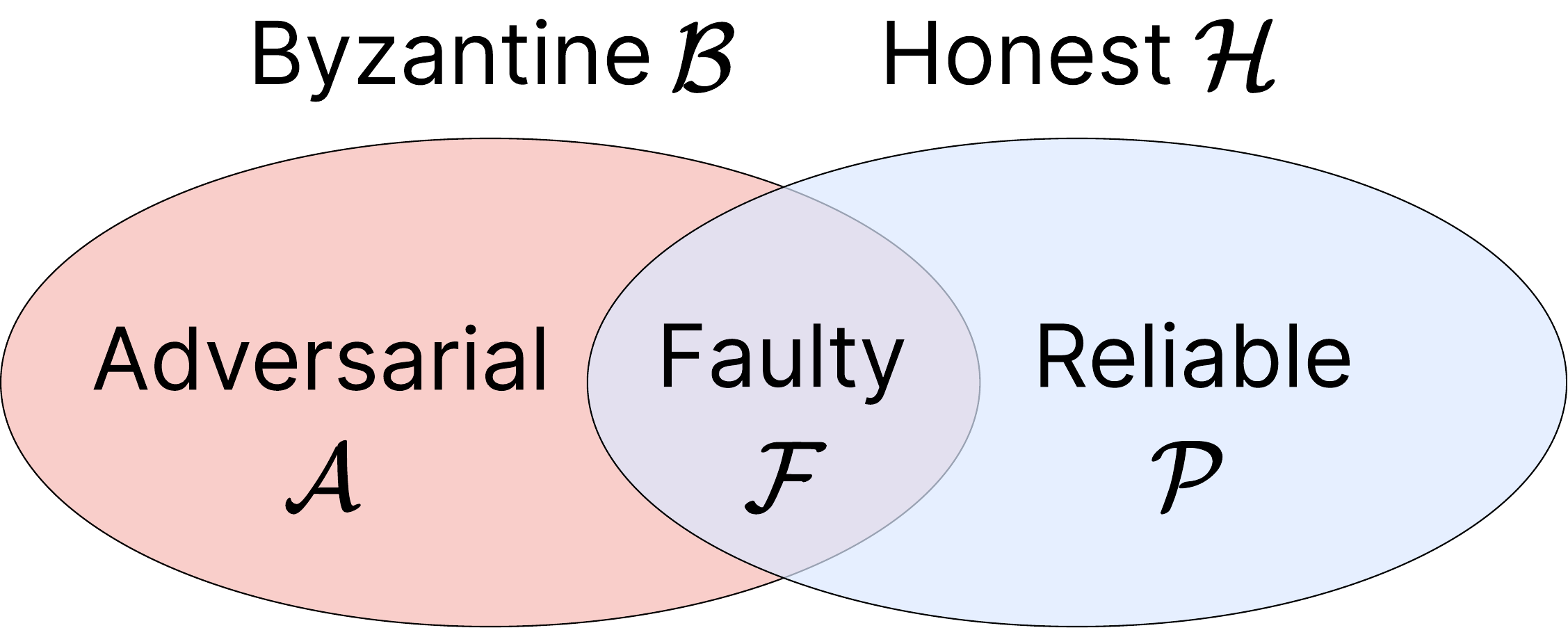}
    \caption{Venn diagram showing the partitioning of agents into reliable agents, faulty agents, and adversarial agents, among honest and Byzantine agents.}
    \label{fig:different_agents}
\end{figure}


\begin{definition}[Byzantine agent]
An agent $i \in \Bcal\subset \Vcal$ is \emph{Byzantine} if, in some round $t$, it deviates from the ideal behavior and transmits unreliable information to neighbors. This may occur as follows:
\begin{itemize}
    \item \emph{Faulty Byzantine:} the agent $b \in \Fcal \subset \Bcal$ follows the prescribed protocol, but produces incorrect or noisy responses due to internal errors, limited capability, or uncertainty.
    \item \emph{Adversarial Byzantine:} the agent $b\in \Acal \subset \Bcal$ does not follow the prescribed protocol by acting arbitrarily, e.g., producing incorrect responses ($r_i^{(t)} \neq y^*$) or sending inconsistent information to different neighbors.
\end{itemize}
\end{definition}

Byzantine agents consist of both faulty agents $\Fcal$ and adversarial agents $\Acal$ such that $\Acal\cap \Fcal=\emptyset$ and $\Acal\cup \Fcal=\Bcal$. Their behaviors may arise from systematic errors (e.g. hallucinations, stochastic errors) as well as adversarial manipulation, which are indistinguishable to non-Byzantine agents.

The remaining agents $\mathcal P:=\Vcal\setminus \Bcal$ are referred to as reliable agents. We assume identities of Byzantine agents are unknown to reliable agents.  Furthermore, we define the set of honest agents as $\Hcal := \Vcal \setminus \Acal$, i.e., non-adversarial agents. The relationships among these agent classes are illustrated in Fig.~\ref{fig:different_agents}. To quantify the extent of Byzantine (both faulty and adversarial) agents, we adopt the $F$-local model~\citep{leblanc2013resilient}:


\begin{definition}[$F$-local~\cite{leblanc2013resilient}]
A set $\Scal\subset\Vcal$ is $F$-local if every node outside $\Scal$ has at most $F$ neighbors in $\Scal$, i.e., $|\Ncal_i\cap\Scal|\leq F$ for all $i\in\Vcal\setminus\Scal$.
\end{definition}

\subsection{Graph Robustness}

We now introduce the notion of $r$-robustness~\citep{leblanc2013resilient}:



\begin{definition}[$r$-robust~\cite{leblanc2013resilient}]
    A graph $\Gcal = (\Vcal,\Ecal)$ is \emph{$r$-robust} if for every pair of nonempty, disjoint subsets $\Scal_1,\Scal_2 \subset \Vcal$, at least one of the subsets contains a node with at least $r$ neighbors outside the subset. That is, there exists a node $i \in \Scal_k$ such that $|\Ncal_i \setminus \Scal_k| \geq r$ for some $k \in \{1, 2\}$.
\end{definition}

Higher robustness levels limit the influence of more Byzantine agents on reliable ones but require denser communication graphs. We now present a useful property of $r$-robustness:
\begin{lemma}
    Given an $r$-robust graph $\Gcal=(\Vcal, \Ecal)$, let $\Gcal'=(\Vcal, \Ecal')$ be a directed graph produced by removing up to $k$ incoming edges of each agent in $\Gcal$. Then, $\Gcal'$ is $(r-k)$-robust. 
    \label{lem:leblanc}
\end{lemma}
The notion of $r$-robustness is a preferable graph-theoretic condition against Byzantine agents, as it can achieve the same level of resilience as classical metrics such as connectivity or minimum degree with sparser connectivity~\citep{leblanc2013resilient, pirani2023}. This sparsity is particularly valuable in LLM-MAS, where each edge in $\mathcal{E}$ corresponds to an LLM-to-LLM message exchange and thus directly contributes to per-round inference cost. We later leverage this notion of robustness to design communication networks that systematically limit Byzantine influence in LLM-MAS.






\begin{figure*}[t]
\centering
\includegraphics[width=\linewidth]{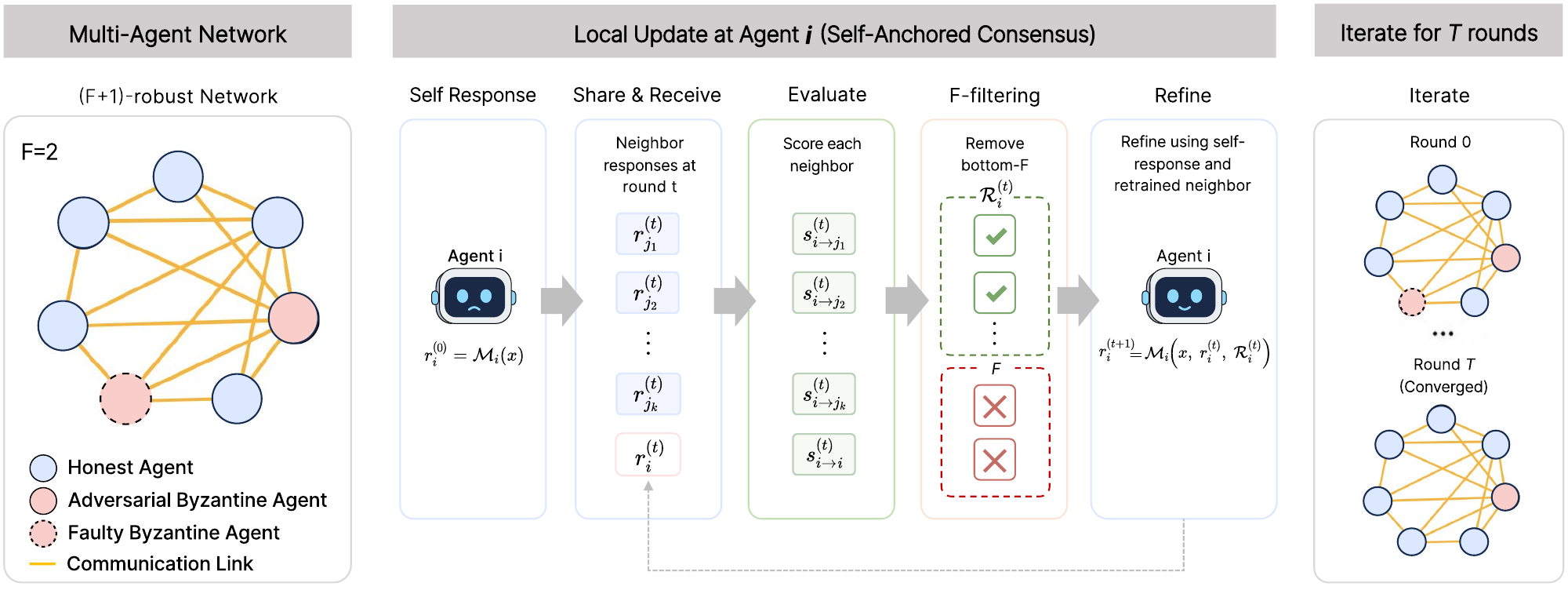}
\caption{Overview of Self-Anchored Consensus (SAC) on an $(F{+}1)$-robust network. At each round, agent $i$ generates its own response, exchanges responses with neighbors, scores them itself, filters out the bottom-$F$ neighbors, and refines its response by re-prompting on the retained set. The procedure is iterated for $T$ rounds.}

\label{fig:method}
\end{figure*}

\subsection{Problem Statement}

Given $n$ LLM agents, a query $x$, and an $F$-local Byzantine set $\Bcal=\Fcal\cup\Acal$, our objective is twofold: (i) ensure that reliable agents $\mathcal P$ remain robust to Byzantine messages, and (ii) enable collaborative refinement that helps all honest agents $\Hcal=\mathcal P\cup\Fcal$ move toward the correct answer in a fully decentralized peer-to-peer network.

We address this in two steps: first, we develop a decentralized protocol where agents iteratively refine responses using neighbor messages while filtering Byzantine influence; second, we use $r$-robustness to characterize communication graphs under which honest agents can reliably refine their responses.

\begin{remark}
Adversarial agents $\Acal$ may arbitrarily deviate from the protocol, whereas faulty agents $\Fcal$ follow the protocol but may produce unreliable outputs. Although both are Byzantine, faulty agents can be viewed as \emph{temporary} Byzantine agents, since iterative refinement may eventually correct their responses and restore reliable behavior.
\end{remark}

\section{Method}
\label{sec:method}


\subsection{\acf{sac}}
\label{sec:confidence}
Here we describe \acf{sac} (visualized in Fig.~\ref{fig:method}), with which each agent $i\in \Vcal$ iteratively refines its response by incorporating trusted neighbor responses while filtering potential Byzantine information through four steps.
Given query $x$, every agent $i\in\Vcal$ first produces an initial response and evaluates its own response:
\begin{equation*}
r_i^{(0)} \;=\; \mathcal{M}_i(x),   \quad   s_{i\to i}^{(0)}=\phi_i(x, r_i^{(0)})\in [0,1],
\end{equation*}
where $\phi_i$ is a scoring function implemented by agent $i$'s own LLM. In our implementation, $\phi_i$ is realized via a verification prompt that asks $\mathcal{M}_i$ to judge the correctness of a given response with respect to $x$, returning a calibrated score in $[0,1]$; the full template is given in~\Cref{app:prompts}.

At each subsequent round $t\geq 0$, every agent $i\in\Vcal$ performs the following steps.

\noindent \textbf{Broadcast responses.} Every agent $i$ broadcasts its current response $r_i^{(t)}$ to its neighbors.

\noindent\textbf{Score neighbors.} Upon receiving neighbor responses $\{r_j^{(t)}\}_{j\in\mathcal{N}_i}$, agent $i$ computes, for each neighbor $j\in\mathcal{N}_i$, an evaluation score
\begin{equation}
s_{i\to j}^{(t)} \;=\; \phi_i\bigl(x,\, r_j^{(t)}\bigr) \,\in\, [0,1].
\label{eq:eval}
\end{equation}

The key distinction from prior work~\cite{zheng2026rethinking} is that $s_{i\to j}^{(t)}$ depends only on agent $i$'s own computation applied to the content of $r_j^{(t)}$, and in particular does \emph{not} rely on any score reported by agent $j$. This removes the possibility for a Byzantine agent to manipulate its own reliability by reporting an inflated self-confidence.

\noindent\textbf{Filter.} Let $\mathcal{L}_i^{(t)}\subseteq\mathcal{N}_i$ be the set of neighbors whose score is strictly below agent $i$'s self-score,
\begin{equation}
\mathcal{L}_i^{(t)} \;=\; \bigl\{\, j\in\mathcal{N}_i \,\bigm|\, s_{i\to j}^{(t)} < s_{i\to i}^{(t)} \,\bigr\}.
\end{equation}
Agent $i$ then discards the $\min(F, |\mathcal{L}_i^{(t)}|)$ lowest-scoring elements of $\mathcal{L}_i^{(t)}$, yielding the retained set
\begin{equation}
\mathcal{R}_i^{(t)} \;=\; \mathcal{N}_i\setminus\mathrm{bottom}_F(\mathcal{L}_i^{(t)}),
\end{equation}
where $\mathrm{bottom}_F(\cdot)$ returns the $F$ lowest-scoring elements of its argument, or the entire argument when $|\mathcal{L}_i^{(t)}|\leq F$.

\noindent\textbf{Refine.} Agent $i$ refines its response by conditioning $\mathcal{M}_i$ on its current response and the retained neighbor responses, weighted by the corresponding receiver-side scores:
\begin{equation}
r_i^{(t+1)} \;=\; \mathcal{M}_i\Bigl( x,\; r_i^{(t)},\; \bigl\{(r_j^{(t)},\, s_{i\to j}^{(t)})\bigr\}_{j\in\mathcal{R}_i^{(t)}} \Bigr).
\label{eq:refine}
\end{equation}
After refinement, agent $i$ re-evaluates its updated response using $\phi_i$ to obtain an updated self-evaluation score.
The prompt template for Eq.~\eqref{eq:refine} instructs $\mathcal{M}_i$ to prefer high-scored responses while retaining the option to keep its own response if none of the retained neighbor responses improves upon it; see~\Cref{app:prompts} for the full template.

After $T$ rounds (a user-specified parameter), the final response of agent $i\in \Vcal$ is $r_i^{(T)}$. The overall procedure of SAC is summarized in Algorithm~\ref{alg:main}.


\begin{algorithm}[h]
\caption{\acf{sac}}
\label{alg:main}
\begin{algorithmic}[1]
\REQUIRE query $x$; adversary bound $F$; number of rounds $T$
\STATE each agent $i$ computes $r_i^{(0)}=\mathcal{M}_i(x)$ and $s_{i \to i}^{(0)}=\phi_i(x,r_i^{(0)})$
\FOR{$t=0,1,\dots,T-1$}
    \STATE each agent $i$ broadcasts $r_i^{(t)}$ to $\mathcal{N}_i$
    \STATE compute $s_{i\to j}^{(t)}$ for all $j\in\mathcal{N}_i$ via Eq.~\eqref{eq:eval}
    \STATE $\mathcal{L}_i^{(t)}\leftarrow\{j\in\mathcal{N}_i:s_{i\to j}^{(t)}<s_{i\to i}^{(t)}\}$
    \STATE $\mathcal{R}_i^{(t)}\leftarrow\mathcal{N}_i\setminus\mathrm{bottom}_F(\mathcal{L}_i^{(t)})$
    \STATE update $r_i^{(t+1)}$ via Eq.~\eqref{eq:refine} and obtain $s_{i \to i}^{(t+1)}$
\ENDFOR
\STATE \textbf{return} $r_i^{(T)}$
\end{algorithmic}
\end{algorithm}

\begin{remark}
Note that~\Cref{alg:main} is fully decentralized; each LLM agent operates using only local information from its neighbors and does not require knowledge of the global network topology $\Gcal$ or the identities of Byzantine agents.
\end{remark}

\begin{remark}
On top of generating its own answer to the given problem, \ac{sac} incurs additional inference cost by independently evaluating neighbors' responses before refinement. With $n$ agents and $T$ communication rounds, \ac{sac} requires $n(1+2T)$ LLM calls: one initialization call and two calls per round for batched scoring and refinement. Since the scoring stage evaluates all received responses in a single prompt, the number of LLM invocations does not directly scale with the neighborhood size, although token usage and inference latency do increase with it.
\label{remark:computational_complexity}
\end{remark}

\subsection{Robust Topological Conditions}

While~\Cref{alg:main} is fully decentralized, its effectiveness depends on the connectivity of the underlying communication graph. Since the filtering step in line~6 of~\ac{sac} removes up to $F$ neighbors per round, an agent with $|\Ncal_i| \leq F$ may have no incoming information after filtering, reducing the algorithm to repeated self-updates. We therefore design the topology to satisfy $(F+1)$-robustness, which ensures two key properties.

\begin{property}
\label{prop:one_honest}
Consider an LLM-MAS connected through a communication graph $\Gcal=(\Vcal,\Ecal)$ with an $F$-local Byzantine set $\Bcal \subset \Vcal$. Let $\Scal_i^{(t)} = \{ \phi_i(x, r_j^{(t)}) \mid j \in \mathcal P \}$ denote the set of evaluation scores assigned by agent $i$ at time $t$ to responses from reliable agents. If $\Gcal$ is $(F+1)$-robust, then for every $k \in \Rcal_i^{(t)}$ and time $t\in \mathbb{Z}_{\geq 0}$, the evaluation score $s_{i\to k}^{(t)}$ is at least as high as the minimum evaluation score assigned to any reliable agent's response, i.e., $s_{i\to k}^{(t)}\geq \min \Scal_i^{(t)}, \ \forall k\in \Rcal_i^{(t)}$.
\end{property}
\begin{proof}
By Lemma 5 of~\citet{leblanc2013resilient}, since $\Gcal$ is $(F+1)$-robust, we have $|\Ncal_i|\geq F+1$ for all $i \in \Vcal$. Since the filtering step of~\ac{sac} removes at most $F$ neighbor indices corresponding to the lowest evaluation scores, and $|\Ncal_i| \ge F+1$, at least one neighbor remains in $\Rcal_i^{(t)}$, i.e., $\Rcal_i^{(t)} \neq \emptyset$. Furthermore, under the $F$-local model, and since at most $F$ indices are removed from the lowest end of the ordered scores, $\Rcal_i^{(t)}$ contains only agent indices whose responses have scores no lower than the minimum score among reliable agents.
\end{proof}

Property~\ref{prop:one_honest} guarantees that, after filtering, each agent retains at least one response whose quality is comparable to that of a reliable agent, \textit{even if the response does not originate from a reliable agent.} This allows honest agents to safely refine their outputs using information of comparable quality to that produced by reliable agents, despite the presence of Byzantine agents. 

\begin{property}
\label{prop:connected}
Consider an LLM-MAS connected through a communication graph $\Gcal=(\Vcal,\Ecal)$ with an $F$-local Byzantine set $\Bcal \subset \Vcal$. Let $\Gcal'=(\Vcal, \Ecal')$ be a directed graph obtained by removing up to any $F$ incoming edges of each agent in $\Gcal$. If $\Gcal$ is $(F+1)$-robust, then $\Gcal'$ contains a rooted-out branching.
\end{property}

Property~\ref{prop:connected} is a direct consequence of Lemma~\ref{lem:leblanc} and Lemma~7 of~\cite{leblanc2013resilient}. It guarantees that $F$-filtering does not disconnect the network, thereby preventing fragmentation into isolated subgroups that may converge to different or incorrect solutions. Therefore, under $(F+1)$-robustness, the subgraph induced by reliable agents remains connected even after filtering out potentially malicious messages.

\begin{table*}[t]
\centering
\small
\begin{adjustbox}{width=1\textwidth}
\begin{tabular}{lllrrrrrrrr}
\toprule
\textbf{Dataset} & \textbf{Method} & \textbf{Topology} & \textbf{IAA} & \textbf{FAA} & \textbf{BFTI} & \textbf{RA} &
\textbf{W IAA}$\to$\textbf{FAA} & \textbf{S IAA}$\to$\textbf{FAA} & \textbf{H-Majority} \\
\midrule
\multirow{6}{*}{MATH} & \multirow{3}{*}{CP-WBFT}
 & MERG          & 52.3\% & 43.6\% & $-8.7\%$  & 35.0\% & 25$\to$58\% & 79$\to$48\% & 38.0\% \\
& & Complete      & 52.7\% & 42.3\% & $-10.4\%$ & 43.0\% & 23$\to$49\% & 81$\to$50\% & 49.0\% \\
& & Erd\H{o}s-R\'enyi & 52.6\% & 35.1\% & $-17.5\%$ & 39.0\% & 25$\to$42\% & 80$\to$41\% & 42.0\% \\
\cmidrule{2-10}
&\multirow{3}{*}{SAC (Ours)}
 & MERG          & 50.1\% & 53.1\% & $+3.0\%$ & 79.0\% & 22$\to$26\% & 77$\to$80\% & 79.0\% \\
& & Complete      & 53.7\% & 54.1\% & $+0.4\%$ & 80.0\% & 26$\to$34\% & 81$\to$78\% & 78.0\% \\
& & Erd\H{o}s-R\'enyi & 53.3\% & 55.0\% & $+1.7\%$ & 79.0\% & 25$\to$38\% & 81$\to$77\% & 78.0\% \\
\midrule
\multirow{6}{*}{Commonsense} &\multirow{3}{*}{CP-WBFT}
 & MERG          & 69.5\% & 37.1\% & $-32.4\%$ & 26.7\% & 77$\to$52\% & 83$\to$39\% & 26.7\% \\
& & Complete   & 69.0\% & 14.3\% & $-54.7\%$ & 16.7\% & 72$\to$17\% & 85$\to$17\% & 16.7\% \\
& & Erd\H{o}s-R\'enyi & 69.0\% & 20.5\% & $-48.5\%$ & 23.3\% & 75$\to$25\% & 83$\to$23\% & 23.3\% \\
\cmidrule{2-10}
&\multirow{3}{*}{SAC (Ours)}
 & MERG          & 68.6\% & 66.7\% & $-1.9\%$  & 76.7\% & 73$\to$77\% & 83$\to$78\% & 80.0\% \\
& & Complete   & 69.5\% & 69.5\% & $+0.0\%$  & 83.3\% & 77$\to$80\% & 83$\to$82\% & 83.3\% \\
& & Erd\H{o}s-R\'enyi & 69.0\% & 66.2\% & $-2.8\%$  & 73.3\% & 73$\to$75\% & 84$\to$78\% & 76.7\% \\
\bottomrule
\end{tabular}
\end{adjustbox}
\caption{Comparison of CP-WBFT and SAC across three communication topologies on
two reasoning benchmarks: 100 Level~4 problems sampled from the Hendrycks MATH test set
and 30 sampled instances from Commonsense170k.
All settings use $n=7$ agents (4 strong honest, 2 weak honest (faulty Byzantine),
1 adversarial Byzantine, $F=3$, $r=4$), with the adversarial Byzantine agent reporting
a falsified self-confidence of 1.0. On MATH we use \texttt{gpt-4o-mini} as the strong model and
\texttt{gpt-3.5-turbo} as the weak model; on Commonsense170k we use \texttt{gpt-5} as the strong model
and \texttt{gpt-4o} as the weak model.}
\label{tab:math500_main_sample}
\end{table*}

\begin{table*}[t]
\centering
\small
\begin{adjustbox}{width=1\textwidth}
\begin{tabular}{llccc cccc}
\toprule
&& \multicolumn{3}{c}{\textbf{MATH}} && \multicolumn{3}{c}{\textbf{Commonsense}}\\
\cmidrule{3-5} \cmidrule{7-9}

\textbf{Method} & \textbf{Round }&
\textbf{MERG (W/S)} & \textbf{Complete (W/S)} & \textbf{Erd\H{o}s-R\'enyi (W/S)} && \textbf{MERG (W/S)} & \textbf{Complete (W/S)} & \textbf{Erd\H{o}s-R\'enyi (W/S)} \\
\midrule
\multirow{2}{*}{CP-WBFT}
 & Init  & 25.0 / 79.0 & 23.0 / 80.8 & 25.0 / 79.5 && 76.7 / 83.3 & 75.0 / 85.0 & 75.0 / 83.3\\
 & Rnd 1-6 & 57.5 / 47.5 & 49.0 / 49.5 & 41.5 / 40.8 && 51.7 / 39.2  & 13.3 / 13.3 & 25.0 / 23.3\\
\midrule
\multirow{7}{*}{SAC (Ours)}
 & Init  & 21.5 / 77.0 & 26.0 / 81.0 & 25.0 / 80.8 && 73.3 / 83.3  & 75.0 / 83.3 & 73.3 / 84.2\\
 & Rnd 1 & 21.5 / 77.0 & 33.5 / 77.8 & 33.0 / 75.8 && 83.3 / 67.5  & 83.3 / 66.7 & 83.3 / 62.5\\
 & Rnd 2 & 27.5 / 80.2 & 34.5 / 80.5 & 34.5 / 77.0 && 73.3 / 79.2  & 73.3 / 80.0 & 70.0 / 77.5\\
 & Rnd 3 & 29.0 / 78.8 & 40.0 / 75.5 & 33.5 / 77.0 && 81.7 / 74.2  & 83.3 / 73.3 & 81.7 / 70.0\\
 & Rnd 4 & 30.5 / 80.8 & 37.0 / 79.2 & 36.0 / 78.2 && 80.0 / 78.3 & 80.0 / 76.7 & 76.7 / 75.0\\
 & Rnd 5 & 27.5 / 80.8 & 37.0 / 80.8 & 39.5 / 75.8 && 80.0 / 71.7  & 76.7 / 70.0 & 78.3 / 67.5\\
 & Rnd 6 & 25.5 / 80.2 & 34.0 / 77.8 & 38.5 / 77.0 && 76.7 / 78.3  & 80.0 / 83.3 & 75.0 / 78.3\\
\bottomrule
\end{tabular}
\end{adjustbox}
\caption{Per-round weak / strong honest accuracy (W/S, in \%) across three
communication topologies on the Hendrycks MATH test set (Level~4, 100 sampled
problems) and Commonsense170k (30 sampled instances). CP-WBFT reports its aggregated result over Rounds 1--6, whereas SAC (Ours)
reports the result at each individual communication round.}
\label{tab:math500_sample_rounds}
\end{table*}

Together, these properties explain why $(F+1)$-robustness is a natural design choice: it provides \emph{local resilience} against Byzantine agents through filtering and \emph{global information-propagation guarantees} across the entire network.

\section{Experimental Results}
\label{sec:experiments}

\begin{table*}[t]
\centering
\small
\begin{adjustbox}{width=1\textwidth}
\begin{tabular}{lllrrrrrrrr}
\toprule
\textbf{Dataset} & \textbf{Method} & \textbf{Topology} & \textbf{IAA} & \textbf{FAA} & \textbf{BFTI} & \textbf{RA} &
\textbf{W IAA}$\to$\textbf{FAA} & \textbf{S IAA}$\to$\textbf{FAA} & \textbf{H-Majority} \\
\midrule
\multirow{6}{*}{MATH-500} & \multirow{3}{*}{CP-WBFT}
 & MERG & 56.5\% & 53.7\% & $-2.8\%$ & 50.0\% & 43$\to$60\% & 69$\to$55\% & 50.0\% \\
& & Complete & 56.6\% & 45.5\% & $-11.1\%$ & 47.2\% & 44$\to$47\% & 68$\to$47\% & 47.2\% \\
& & Erd\H{o}s-R\'enyi & 57.2\% & 47.9\% & $-9.3\%$ & 49.6\% & 45$\to$50\% & 69$\to$50\% & 49.6\% \\
\cmidrule{2-10}
& \multirow{3}{*}{SAC (Ours)}
 & MERG & {56.1\%} & {62.6\%} & {$+6.5\%$} & {72.4\%} & 43$\to$60\% & 68$\to$71\% & {72.6\%} \\
& & Complete & {55.7\%} & {62.1\%} & {$+6.4\%$} & {74.4\%} & 42$\to$55\% & 68$\to$72\% & {74.6\%} \\
& & Erd\H{o}s-R\'enyi & {56.5\%} & {61.9\%} & {$+5.4\%$} & {72.6\%} & 44$\to$58\% & 68$\to$71\% & {73.2\%} \\
\midrule
\multirow{6}{*}{Commonsense} & \multirow{3}{*}{CP-WBFT}
 & MERG & 46.6\% & 39.6\% & $-7.0\%$ & 38.8\% & 50$\to$46\% & 52$\to$42\% & 38.8\% \\
& & Complete & 46.7\% & 34.7\% & $-12.0\%$ & 38.2\% & 50$\to$37\% & 53$\to$38\% & 38.2\% \\
& & Erd\H{o}s-R\'enyi & 46.6\% & 34.8\% & $-11.8\%$ & 38.3\% & 50$\to$38\% & 52$\to$38\% & 38.3\% \\
\cmidrule{2-10}
& \multirow{3}{*}{SAC (Ours)}
 & MERG & {46.8\%} & {47.9\%} & {$+1.2\%$} & {55.4\%} & 50$\to$49\% & 53$\to$55\% & {55.4\%} \\
& & Complete & {46.9\%} & {48.9\%} & {$+2.0\%$} & {56.2\%} & 50$\to$50\% & 53$\to$56\% & {56.3\%} \\
& & Erd\H{o}s-R\'enyi & {46.6\%} & {47.8\%} & {$+1.1\%$} & {55.3\%} & 50$\to$49\% & 52$\to$55\% & {55.5\%} \\
\bottomrule
\end{tabular}
\end{adjustbox}
\caption{Comparison of CP-WBFT and SAC across three communication topologies on
the full MATH-500 benchmark and five commonsense reasoning benchmarks
(ARC-Challenge, HellaSwag, BoolQ, OpenBookQA, and RTE).
All results are averaged over the full evaluation sets. All settings use $n=7$
agents: 4 strong honest agents, 2 weak honest agents, and 1 adversarial
Byzantine agent. Strong/weak agents use \texttt{Qwen3-4B}/
\texttt{Qwen2.5-1.5B-Instruct}.}
\label{tab:math500_main}
\end{table*}

\begin{table*}[t]
\centering
\small
\begin{adjustbox}{width=1\textwidth}
\begin{tabular}{llccc cccc}
\toprule
&& \multicolumn{3}{c}{\textbf{MATH-500}} &&
\multicolumn{3}{c}{\textbf{Commonsense}}\\
\cmidrule{3-5} \cmidrule{7-9}
\textbf{Method} & \textbf{Round} &
\textbf{MERG (W/S)} &
\textbf{Complete (W/S)} &
\textbf{Erd\H{o}s-R\'enyi (W/S)} &&
\textbf{MERG (W/S)} &
\textbf{Complete (W/S)} &
\textbf{Erd\H{o}s-R\'enyi (W/S)} \\
\midrule
\multirow{2}{*}{CP-WBFT}
 & Init
 & 42.7 / 68.8 & 43.9 / 68.3 & 44.6 / 69.0
 && 49.7 / 52.3 & 50.0 / 52.6 & 50.2 / 52.2 \\
 & Rnds 1--6
 & 60.4 / 55.0 & 47.2 / 47.4 & 49.8 / 50.2
 && 45.6 / 42.2 & 36.5 / 38.2 & 37.5 / 37.9 \\
\midrule
\multirow{7}{*}{SAC (Ours)}
 & Init
 & 42.9 / 67.9 & 41.5 / 68.0 & 44.2 / 68.0
 && 50.0 / 52.6 & 50.2 / 52.7 & 49.7 / 52.5 \\
 & Rnd 1
 & 42.9 / 67.9 & 41.5 / 68.0 & 44.2 / 68.0
 && 48.0 / 53.1 & 49.0 / 54.4 & 47.8 / 53.2 \\
 & Rnd 2
 & 55.8 / 69.2 & 50.0 / 71.5 & 53.4 / 70.7
 && 48.8 / 55.6 & 49.5 / 55.9 & 48.8 / 54.9 \\
 & Rnd 3
 & 59.2 / 69.5 & 53.0 / 71.1 & 55.8 / 71.6
 && 48.1 / 54.0 & 49.1 / 54.4 & 48.0 / 53.5 \\
 & Rnd 4
 & 59.5 / 69.9 & 53.0 / 71.5 & 57.0 / 71.2
 && 48.7 / 55.5 & 49.6 / 56.1 & 48.8 / 55.0 \\
 & Rnd 5
 & 59.1 / 70.2 & 53.5 / 71.3 & 56.6 / 71.3
 && 48.2 / 54.0 & 49.1 / 54.2 & 48.0 / 53.5 \\
 & Rnd 6
 & 60.0 / 70.6 & 53.8 / 71.1 & 56.7 / 71.2
 && 48.6 / 55.3 & 49.6 / 56.4 & 48.6 / 55.0 \\
\bottomrule
\end{tabular}
\end{adjustbox}
\caption{Per-round weak / strong honest accuracy (W/S, in \%) across three
communication topologies on the full MATH-500 benchmark and five commonsense
reasoning benchmarks (ARC-Challenge, HellaSwag, BoolQ, OpenBookQA, and RTE).}
\label{tab:per_round_math500_rounds}
\end{table*}

\subsection{Evaluation Metrics}

Let $y^\star$ denote the ground-truth answer to query $x \sim \mathcal{D}$. We partition honest agents $\Hcal$ into strong and weak groups, where strong and weak agents are treated as reliable and faulty (non-adversarial Byzantine) agents, respectively. Following~\citet{zheng2026rethinking}, we report the following metrics (higher is better). We set communication round to be $T=6$ for all experiments.

\noindent\textbf{Initial / Final Agent Accuracy (IAA / FAA).}
Average per-agent accuracy before communication ($t=0$) and after $T$ consensus rounds ($t=T$).

\noindent\textbf{Byzantine Fault Tolerance Improvement (BFTI).}
$\mathrm{BFTI} = \mathrm{FAA} - \mathrm{IAA}$ measures the gain or degradation introduced by the consensus protocol.

\noindent\textbf{Round-level Accuracy (RA).}
The fraction of queries for which the majority vote across all agents in $\Vcal$ at round $T$ matches $y^\star$.

\noindent\textbf{Group-wise Accuracy (W\_IAA->FAA, S\_IAA->FAA).}
Accuracy change of the weak and strong honest groups before and after consensus.

\noindent\textbf{Honest Majority (H-Majority).}
The fraction of queries for which the majority vote among honest agents $\Hcal$ at round $T$ matches $y^\star$.

\begin{table*}[t]
\centering
\small
\begin{adjustbox}{width=1\textwidth}
\begin{tabular}{lllrrrrrrrr}
\toprule
\textbf{Dataset} & \textbf{Method} & \textbf{Topology} & \textbf{IAA} & \textbf{FAA} & \textbf{BFTI} & \textbf{RA} &
\textbf{W IAA}$\to$\textbf{FAA} & \textbf{S IAA}$\to$\textbf{FAA} & \textbf{H-Majority} \\
\midrule
\multirow{6}{*}{AIME 2025}
& \multirow{3}{*}{CP-WBFT}
& MERG
& 60.5\% & 21.0\% & $-39.5\%$ & 0.0\%
& 63$\to$37\% & 74$\to$18\% & 0.0\% \\
& & Complete
& 61.9\% & 0.0\% & $-61.9\%$ & 0.0\%
& 62$\to$0\% & 78$\to$0\% & 0.0\% \\
& & Erd\H{o}s-R\'enyi
& 61.0\% & 0.0\% & $-61.0\%$ & 0.0\%
& 62$\to$0\% & 76$\to$0\% & 0.0\% \\
\cmidrule{2-10}
& \multirow{3}{*}{SAC (Ours)}
& MERG
& {62.4\%} & {61.4\%} & {$-1.0\%$} & {80.0\%}
& 63$\to$78\% & 78$\to$68\% & {83.3\%} \\
& & Complete
& {61.4\%} & {58.6\%} & {$-2.8\%$} & {76.7\%}
& 67$\to$72\% & 74$\to$67\% & {76.7\%} \\
& & Erd\H{o}s-R\'enyi
& {63.3\%} & {61.4\%} & {$-1.9\%$} & {83.3\%}
& 72$\to$75\% & 75$\to$70\% & {80.0\%} \\
\bottomrule
\end{tabular}
\end{adjustbox}
\caption{Comparison of CP-WBFT and SAC across three communication topologies on
the full AIME 2025 benchmark. All settings use $n=7$ agents: 4 strong
honest agents, 2 weak honest agents, and 1 adversarial Byzantine agent,
with $T=6$ communication rounds. Strong agents use
\texttt{Qwen3-4B-Thinking-2507} with a 32K-token generation budget, while weak
agents use \texttt{Qwen3-4B}.}
\label{tab:aime2025_main}
\end{table*}

\subsection{Experimental Setups}




\noindent\textbf{Datasets.}
While \citet{zheng2026rethinking} evaluates their method on $10$ hand-curated GSM8K problems, we provide a broader evaluation with more challenging mathematical and commonsense reasoning benchmarks across multiple model settings. For closed-weight LLM experiments, repeated multi-agent API calls incur substantial cost; we thus evaluate on sampled subsets consisting of $100$ Level~4 problems from the Hendrycks MATH test set~\citep{MATHall} and $30$ instances from Commonsense170k~\citep{hu2023llm}. For open-weight LLM experiments, we evaluate on the full MATH-500 benchmark~\citep{MATHall}, which is distinct from the sampled Hendrycks MATH Level~4 setting used for closed-weight models. We further evaluate on the full evaluation sets of five commonsense reasoning benchmarks: ARC-C~\citep{clark2018arc}, HellaSwag~\citep{zellers2019hellaswag}, BoolQ~\citep{clark2019boolq}, OBQA~\citep{mihaylov2018openbookqa}, and RTE~\citep{rte}. Finally, for reasoning LLM experiments, we evaluate on the full AIME 2025 and HMMT-25 benchmarks. The sampled problem lists for the closed-weight experiments and additional dataset details are provided in Appendix~\ref{app:dataset}. 


\noindent\textbf{Network.}
Each network contains $n=7$ agents: $4$ strong honest, $2$ weak honest (faulty Byzantine), and $1$ adversarial Byzantine. We set $F=3$ and evaluate three $(F{+}1)$-robust topologies: $\gamma$-MERG~\citep{lee2026minimal}, complete~\citep{leblanc2013resilient}, and Erd\H{o}s--R\'enyi~\citep{erdHos1961strength,zhang2015notion}. More details are provided in Appendix~\ref{sec:appendix_graph}.


\noindent\textbf{Models.}
We evaluate SAC across diverse model families and capability levels.
The weak/strong model pairs are
\texttt{gpt-3.5-turbo} and \texttt{gpt-4o-mini} for Hendrycks MATH,
\texttt{gpt-4o} and \texttt{gpt-5} for closed-weight commonsense reasoning,
\texttt{Qwen2.5-1.5B-Instruct} and \texttt{Qwen3-4B} for open-weight experiments,
and \texttt{Qwen3-4B} and \texttt{Qwen3-4B-Thinking-2507} for AIME 2025.
For AIME 2025, the strong model uses a 32K-token generation budget.
We additionally evaluate \texttt{GPT-5.6-sol} as the strong model on HMMT-25.

\noindent\textbf{Byzantine model.}
Adversarial Byzantine agents return out-of-distribution weak-model responses with falsified self-confidence scores of $1.0$, simulating dishonest-confidence attacks against CP-WBFT while remaining indistinguishable from reliable agents. Faulty Byzantine agents follow the same protocol as reliable agents.

\begin{table*}[t]
\centering
\small
\begin{adjustbox}{width=1\textwidth}
\begin{tabular}{lllrrrrrrr}
\toprule
\textbf{Dataset} & \textbf{Method} & \textbf{Topology} &
\textbf{Strong Agent ACC} & \textbf{IAA} & \textbf{FAA} &
\textbf{BFTI} & \textbf{RA} &
\textbf{Strong I}$\to$\textbf{F} & \textbf{H-Majority} \\
\midrule
\multirow{4}{*}{HMMT-25}
& \multirow{2}{*}{CP-WBFT}
& MERG
& 100.0\% & 57.1\% & 5.7\% & $-51.4\%$ & 10.0\%
& 100$\to$10\% & 10.0\% \\
& & Complete
& 100.0\% & 57.1\% & 5.7\% & $-51.4\%$ & 10.0\%
& 100$\to$10\% & 10.0\% \\
\cmidrule{2-10}
& \multirow{2}{*}{SAC (Ours)}
& MERG
& {100.0\%} & {57.1\%} & {53.3\%}
& {$-3.8\%$} & {93.3\%}
& {100$\to$93\%} & {93.3\%} \\
& & Complete
& {100.0\%} & {57.1\%} & {53.3\%}
& {$-3.8\%$} & {93.3\%}
& {100$\to$93\%} & {93.3\%} \\
\bottomrule
\end{tabular}
\end{adjustbox}
\caption{
Comparison of CP-WBFT and SAC across two communication topologies on the full HMMT-25 benchmark. All settings use $n=7$ agents: four \texttt{GPT-5.6-sol} strong honest agents and three adversarial Byzantine agents, with no weak honest agents. 
}
\label{tab:gpt5}
\end{table*}

\begin{table*}[t]
\centering
\scriptsize
\begin{adjustbox}{width=1\textwidth}
\begin{tabular}{@{}ll *{5}{ccc}@{}}
\toprule
 & & \multicolumn{3}{c}{\textbf{IAA (\%)}} & \multicolumn{3}{c}{\textbf{FAA (\%)}} & \multicolumn{3}{c}{\textbf{BFTI (\%)}} & \multicolumn{3}{c}{\textbf{RA (\%)}} & \multicolumn{3}{c}{\textbf{H-Majority (\%)}} \\
\cmidrule(lr){3-5}\cmidrule(lr){6-8}\cmidrule(lr){9-11}\cmidrule(lr){12-14}\cmidrule(lr){15-17}
&
 & b=1 & b=2 & b=3
 & b=1 & b=2 & b=3
 & b=1 & b=2 & b=3
 & b=1 & b=2 & b=3
 & b=1 & b=2 & b=3 \\
\textbf{Method} & \textbf{Topology} & w=2 & w=1 & w=0
   & w=2 & w=1 & w=0
   & w=2 & w=1 & w=0
   & w=2 & w=1 & w=0
   & w=2 & w=1 & w=0 \\
\midrule
\multirow{3}{*}{CP-WBFT}
 & MERG          & 55.2 & 48.1 & 44.3 & 44.8 & 31.4 & 29.0 & $-10.4$ & $-16.7$ & $-15.3$ & 30.0 & 40.0 & 50.0 & 40.0 & 43.3 & 50.0 \\
 & Complete & 56.2 & 51.9 & 49.0 & 31.4 & 30.5 & 34.3 & $-24.8$ & $-21.4$ & $-14.7$ & 33.3 & 43.3 & 60.0 & 36.7 & 43.3 & 60.0 \\
 & Erd\H{o}s-R\'enyi & 56.2 & 52.9 & 48.1 & 40.0 & 31.4 & 29.5 & $-16.2$ & $-21.5$ & $-18.6$ & 43.3 & 40.0 & 50.0 & 46.7 & 40.0 & 53.3 \\
\midrule
\multirow{3}{*}{SAC (Ours)}
 & MERG          & 54.8 & 49.5 & 47.1 & 62.4 & 52.9 & 48.1 & $+7.6$  & $+3.4$  & $+1.0$  & 90.0 & 83.3 & 86.7 & 96.7 & 83.3 & 83.3 \\
 & Complete & 55.7 & 51.0 & 47.6 & 59.5 & 54.8 & 52.4 & $+3.8$  & $+3.8$  & $+4.8$  & 83.3 & 86.7 & 93.3 & 86.7 & 83.3 & 93.3 \\
 & Erd\H{o}s-R\'enyi & 54.3 & 51.4 & 46.7 & 60.5 & 57.1 & 47.1 & $+6.2$  & $+5.7$  & $+0.4$  & 80.0 & 93.3 & 83.3 & 86.7 & 90.0 & 86.7 \\
\bottomrule
\end{tabular}
\end{adjustbox}
\caption{Ablation over byzantine-to-weak composition at $n{=}7$ on the Hendrycks MATH test set (30 sampled Level~4 problems). The byzantine agent count ($b$) and weak honest count ($w$) are varied while the strong honest count is fixed at $4$ ($F{=}3$, $T{=}6$, $r{=}4$). Strong/weak agents use \texttt{gpt-4o-mini}/\texttt{gpt-3.5-turbo}.}
\label{tab:bw-ablation}
\end{table*}

\begin{table*}[t]
\centering
\footnotesize
\begin{tabular}{@{}llcccccc@{}}
\toprule
\textbf{Method} & \textbf{Topology} & \textbf{$r$ }& \textbf{IAA (\%) }& \textbf{FAA (\%)} & \textbf{BFTI (\%)} & \textbf{RA (\%)} & \textbf{H-Majority (\%)} \\
\midrule
\multirow{4}{*}{CP-WBFT}
 & MERG          & 5 & 63.7 & 37.0 & $-26.7$ & 20.0 & 23.3 \\
 & Complete & 5 & 61.9 & 29.6 & $-32.3$ & 33.3 & 33.3 \\
 & Preferential  & 4 & 64.8 & 44.1 & $-20.7$ & 26.7 & 26.7 \\
 & Erd\H{o}s-R\'enyi & 4 & 63.0 & 50.4 & $-12.6$ & 36.7 & 43.3 \\
\midrule
\multirow{4}{*}{SAC (Ours)}
 & MERG          & 5 & 62.6 & 64.8 & $+2.2$  & 90.0 & 90.0 \\
 & Complete & 5 & 60.4 & 62.2 & $+1.8$  & 86.7 & 86.7 \\
 & Preferential  & 4 & 62.6 & 65.6 & $+3.0$  & 90.0 & 90.0 \\
 & Erd\H{o}s-R\'enyi & 4 & 62.6 & 60.4 & $-2.2$  & 86.7 & 86.7 \\
\bottomrule
\end{tabular}
\caption{Scalability to larger networks ($n{=}9$ agents) on the Hendrycks MATH test set (Level~4, 30 sampled problems). Configuration: 6 strong honest, 2 weak honest, and 1 adversarial Byzantine ($F{=}3$, $T{=}6$). Strong/weak agents use \texttt{gpt-4o-mini}/\texttt{gpt-3.5-turbo}.}
\label{tab:n9-main}
\end{table*}

\subsection{Experimental Results}

\subsubsection{Closed-Weight LLMs}
We first test our method with closed-weight LLMs on 100 Level~4 problems sampled from the Hendrycks MATH test set
and 30 sampled instances from Commonsense170k, and the results are shown in Tables~\ref{tab:math500_main_sample} and~\ref{tab:math500_sample_rounds}.

\noindent\textbf{MATH.} We use \texttt{gpt-3.5-turbo} and \texttt{gpt-4o-mini} as the weak and strong models, respectively.
CP-WBFT consistently yields negative BFTI across all topologies ($-8.7\%$ to $-17.5\%$), with strong-agent accuracy decreasing after the first communication round while weak-agent accuracy increases. In contrast, \ac{sac} achieves non-negative BFTI ($+0.4\%$ to $+3.0\%$), preserves strong-agent capability, and maintains stable per-round behavior. H-Majority reaches $78$--$79\%$ across all topologies.

\noindent\textbf{Commonsense reasoning.}
A similar trend can be observed on Commonsense170k, where we use \texttt{gpt-4o} and \texttt{gpt-5} as the weak and strong models. CP-WBFT suffers severe degradation under Byzantine influence ($-54.7\%$ to $-32.4\%$ BFTI), with both weak and strong agents collapsing immediately after Rnd~1. In contrast, \ac{sac} maintains near-zero degradation ($-2.8\%$ to $+0.0\%$ BFTI), preserving the performance of strong and weak agents. 


\subsubsection{Open-Weight LLMs}
\label{openweight}
We now evaluate our method with open-weight LLMs, using \texttt{Qwen2.5-1.5B-Instruct} and \texttt{Qwen3-4B} as the weak and strong models, respectively. The results in Tables~\ref{tab:math500_main} and~\ref{tab:per_round_math500_rounds} show a similar trend as that of the closed-weight LLMs. CP-WBFT consistently yields negative BFTI across all topologies ($-2.8\%$ to $-12.0\%$), with both weak and strong groups suffering from accuracy loss after the first communication round. In contrast, SAC consistently achieves positive BFTI ($+1.1\%$ to $+6.5\%$) with stable or improving per-round behavior. SAC achieves H-Majority of $55.4$-$74.6\%$, compared to $38.2$-$50.0\%$ for CP-WBFT. Across all topologies, SAC preserves strong-agent capability while preventing collapse of weaker agents.

\subsection{Reasoning LLMs}
We further evaluate the methods with reasoning LLMs on the AIME 2025 benchmark. We use \texttt{Qwen3-4B-Thinking-2507} with a 32K-token generation budget and \texttt{Qwen3-4B} as the strong and weak agents. As shown in Table~\ref{tab:aime2025_main}, \ac{sac} greatly outperforms CP-WBFT with substantially higher BFTI and H-Majority across all topologies.

Similar performance trend can be seen even with the state-of-the-art model \texttt{GPT-5.6-sol} on the HMMT-25 benchmark. For this experiment, we have four strong honest agents and three adversarial agents. Table~\ref{tab:gpt5} shows that \ac{sac} better preserves the performance of honest agents: it achieves a BFTI of only $-3.8\%$, compared with $-51.4\%$ for CP-WBFT, while achieving $93.3\%$ H-Majority compared with only $10.0\%$ for CP-WBFT.

\subsection{Ablation Study}

We conduct ablation studies using 30 sampled Level 4 problems from the Hendrycks MATH test set to examine the effects of Byzantine-agent composition and network size $n$ on \ac{sac}.

\noindent\textbf{Varying byzantine-to-weak composition.}
Table~\ref{tab:bw-ablation} fixes the strong honest count at $4$ and varies $(b,w)\in\{(1,2),(2,1),(3,0)\}$ at $n{=}7$. CP-WBFT yields negative BFTI across all settings, whereas SAC maintains non-negative BFTI in nearly all cases with high H-Majority ($83.3$-$96.7\%$). Performance does not degrade monotonically as $b$ increases, consistent with $(F{+}1)$-robustness limiting adversarial influence under the $F$-local assumption.

\noindent\textbf{Scaling to larger networks.}
Table~\ref{tab:n9-main} scales the system to $n{=}9$ with a $6/2/1$ split at the same adversary bound $F{=}3$. While CP-WBFT still degrades substantially, SAC maintains positive BFTI on three of four topologies and H-Majority of $86.7$--$90.0\%$, including on a $4$-robust preferential-attachment graph. These results suggest that our method scales to larger networks without retuning.

\section{Discussion}

\noindent\textbf{Self-reported confidence creates a point of vulnerability.}
Across all evaluated topologies and benchmarks, CP-WBFT with adversarial agents with falsified self-confidence consistently exhibits rapid and significant performance degradation. This suggests that manipulable trust signals themselves become a critical vulnerability.

\noindent\textbf{Receiver-side evaluation contains Byzantine influence.}
By construction, \ac{sac} contains the impact of Byzantine agents through receiver-side evaluation. This allows \ac{sac} to achieve non-negative BFTI in nearly all settings while achieving substantially higher H-Majority than CP-WBFT.

\noindent\textbf{Strong agents are preserved while weak agents are lifted.}
CP-WBFT may improve weak-agent accuracy, but by dragging strong agents toward incorrect responses, degrading overall reliability. In contrast, SAC preserves strong-agent capability with only minor degradation while improving weak agents across all topologies in general. These results suggest that SAC acts as a \emph{protective filtering mechanism} against Byzantine faults.

\section{Conclusion}
We studied robust decision-making in decentralized LLM-MAS under Byzantine faults. We proposed \acf{sac}, a fully decentralized filter-and-refine protocol with receiver-side evaluation, and established $(F{+}1)$-robustness conditions under which honest agents reliably refine their responses despite Byzantine influence. Experiments on mathematical and commonsense reasoning benchmarks show that \ac{sac} consistently maintains robustness across diverse topologies, whereas prior methods suffer a collapse in performance under Byzantine attacks.


\section*{Limitations}

One main limitation of \ac{sac} is that its practical robustness relies on the agents' base capabilities, as effective filtering and refinement require both (i) a reliable scoring function $\phi_i$ for identifying correct responses and (ii) the ability to effectively integrate neighbor information. Thus, a key failure mode could arise when a faulty agent is overconfident and assigns low scores to correct neighbor responses, thereby discarding useful information and remaining trapped in error despite reliable neighbors. A complete characterization of how model capabilities affect performance remains future work.

In addition, our evaluation is limited to networks of $n \in \{7, 9\}$, a single adversarial behavior (cached weak-model replay with falsified confidence), and single-turn responses. Larger and dynamic networks, multi-turn reasoning, and tool use are outside the present scope.

Lastly, we note that stronger robustness can improve resilience but this requires denser network, thereby increasing computational and communication costs, as discussed in Remark~\ref{remark:computational_complexity}. This highlights a resilience-efficiency trade-off, which we plan to investigate in future work.

\section*{Ethical considerations}
This work studies robustness in decentralized multi-agent LLM systems under Byzantine faults. Our experiments are conducted using publicly available benchmark datasets and commercial or open-weight language models. No human subjects, personal data, or sensitive information are involved.

Although our work considers adversarial agents that provide misleading responses or falsified confidence scores, these settings are studied solely for robustness evaluation and reliability analysis. We hope this work contributes to the development of more reliable and trustworthy collaborative LLM systems.

\bibliography{custom}

\appendix

\clearpage
\newpage
\section*{Appendix}
\section{Graph Topologies}
\label{sec:appendix_graph}
We evaluate our method on four graph topologies which we describe below:

\begin{itemize}
\item $\gamma$-MERGs~\cite{lee2026minimal}: These graphs achieve maximum robustness, i.e., $r=\lceil n/2 \rceil$, for a given number of nodes $n$, while using the minimal set of edges possible. We distinguish two cases for constructing a $\gamma$-MERG:

\begin{itemize}
    \item \textbf{If $n$ is odd:} Let $\mathcal{X} \subset \mathcal{V}$ be a set of $\gamma + 1$ nodes forming a complete subgraph. Connect each of the remaining $\gamma - 1$ nodes to any $\gamma$ nodes in $\mathcal{X}$.

    \item \textbf{If $n$ is even:} Let $\mathcal{X} \subset \mathcal{V}$ be a set of $\gamma$ nodes, each adjacent to all other nodes in $\mathcal{V}$. Then select $\left\lceil \frac{\gamma - 2}{2} \right\rceil$ disjoint pairs of nodes in $\mathcal{X}$ and remove the edge between each pair.
\end{itemize}

By Theorem~1 of~\cite{lee2026minimal}, the resulting graph is $\lceil n/2 \rceil$-robust.

    \item Complete graphs: In this case, every agent is connected to every other agent. By Lemma~4 of~\citet{leblanc2013resilient}, a complete graph with $n$ nodes is $r=\lceil n/2 \rceil$-robust.

    \item Preferential-attachment graphs: To construct $r$-robust graphs, we first form a complete graph with $2r-1$ nodes. By Lemma~4 of~\citet{leblanc2013resilient}, this subgraph is $r$-robust. The remaining $n-2r+1$ nodes are then each connected to any $r$ nodes in the original set of nodes in the complete graph. By Theorem~5 of~\citet{leblanc2013resilient}, the resulting graph is then $r$-robust.

    \item Erd\H{o}s--R\'enyi random graphs~\cite{erdHos1961strength}: These graphs are parameterized by the number of nodes $n$ and an edge probability $p$. The presence of each edge is independent of every other edge, and each edge is included independently with probability $p$. As shown in Theorem~3 of~\citet{zhang2015notion}, choosing 
    $p=\frac{\ln n + (r-1)\ln \ln n}{n}$
    ensures that the graph is $r$-robust with high probability as $n\to\infty$. We sample graphs using this probability and verify $r$-robustness prior to simulation using the method from~\citet{usevitch2020determining}.
\end{itemize}

\section{Prompt Templates}
\label{app:prompts}

We instantiate the receiver-side scoring function $\phi_i$ (Eq.~\eqref{eq:eval}) and the refine operator (Eq.~\eqref{eq:refine}) using the prompt templates below. We provide the prompts used for the sampled Hendrycks MATH (Level~4) experiments. For Commonsense170k, we use the same overall structure with minor modifications to the task instruction and answer format, as described in~\Cref{app:prompts:commonsense}.

\subsection{Scoring Prompt ($\phi_i$)}
\label{app:prompts:score}

The scoring prompt instructs $\mathcal{M}_i$ to independently solve the query $x$, compare its own reasoning against $r_j^{(t)}$, and finally output a calibrated confidence score in $[0,1]$. This procedure makes $s_{i\to j}^{(t)}$ a receiver-side evaluation: the score depends solely on $\mathcal{M}_i$'s internal reasoning applied to the content of $r_j^{(t)}$, without relying on any confidence reported by agent $j$.

\begin{tcolorbox}[colback=gray!5,colframe=black!50,boxrule=0.4pt,fontupper=\small\ttfamily,breakable]
\textbf{System:} You are an expert evaluator. Respond using only a single number.\\[4pt]
\textbf{User:}\\
You are evaluating the correctness of a proposed answer to a mathematical reasoning problem.\\[2pt]
Assume that up to $f$ responses may be malicious or unreliable.\\[2pt]
Reason carefully before producing the final score.\\[4pt]
Problem: \{$x$\}\\
Proposed answer: \{$r_j^{(t)}$\}\\[4pt]
Step 1: Solve the problem independently step by step.\\
Step 2: Evaluate the proposed answer.\\
Step 3: Output ONLY a single decimal number between 0.0 and 1.0.\\
0.0 = clearly incorrect, 0.5 = uncertain, 1.0 = clearly correct.\\
Score:
\end{tcolorbox}

\noindent The first numeric token in $\mathcal{M}_i$'s response is extracted using the regular expression \texttt{([0-9]+(?:\textbackslash.[0-9]+)?)} and clipped to the interval $[0,1]$. If parsing fails or the LLM returns an invalid response, the score defaults to $s_{i\to j}^{(t)}=0.5$, treating the response as maximally uncertain. The self-score $s_{i\to i}^{(t)}=\phi_i(x, r_i^{(t)})$ is computed using the same prompt with $r_j^{(t)}$ replaced by $r_i^{(t)}$.

\subsection{Refine Prompt (Eq.~\eqref{eq:refine})}
\label{app:prompts:refine}

The refine prompt provides $\mathcal{M}_i$ with its current response $r_i^{(t)}$ together with the retained neighbor responses $\{(r_j^{(t)}, s_{i\to j}^{(t)})\}_{j\in\Rcal_i^{(t)}}$, sorted in descending order of reliability score. The instruction
\emph{``Prefer answers with higher reliability scores, but retain your current answer if you believe it is correct''}
serves as the natural-language realization of Eq.~\eqref{eq:refine}, encouraging $\mathcal{M}_i$ to prioritize highly scored neighbor responses while preserving its own response as an anchor when appropriate.

\begin{tcolorbox}[colback=gray!5,colframe=black!50,boxrule=0.4pt,fontupper=\small\ttfamily,breakable]
\textbf{System:} You are a helpful assistant. Follow the required output format exactly.\\[4pt]
\textbf{User:}\\
You are solving a mathematical reasoning problem. Review the information below and provide your best answer.\\[2pt]
Assume that up to $f$ responses may be malicious or unreliable.\\[2pt]
Reason carefully before producing the final answer.\\[4pt]
Problem: \{$x$\}\\[2pt]
Your current answer: \{$r_i^{(t)}$\}\\[4pt]
Other agents' answers with reliability scores (higher indicates greater reliability):\\
\quad - Answer: \{$r_{j_1}^{(t)}$\}\quad(reliability score: \{$s_{i\to j_1}^{(t)}$\})\\
\quad - Answer: \{$r_{j_2}^{(t)}$\}\quad(reliability score: \{$s_{i\to j_2}^{(t)}$\})\\
\quad \dots\quad(sorted in descending order of score, $j_k\in\Rcal_i^{(t)}$)\\[4pt]
Prefer answers with higher reliability scores, but retain your current answer if you believe it is correct.\\
Respond using the format:\\
Answer: [your numerical answer]
\end{tcolorbox}

\noindent Reliability scores are formatted to two decimal places. The retained set $\Rcal_i^{(t)}$ is computed prior to rendering the prompt, so neighbors removed by $\mathrm{bottom}_F(\cdot)$ never appear in the refinement stage. If $\Rcal_i^{(t)}=\emptyset$, the refinement step is skipped and $r_i^{(t+1)} = r_i^{(t)}$. The final numerical answer is extracted using a dataset-specific parser; if extraction fails, the update again defaults to $r_i^{(t+1)} = r_i^{(t)}$.

\subsection{Commonsense170k Prompt Adaptation}
\label{app:prompts:commonsense}

Since the sampled Commonsense170k instances span heterogeneous formats, including true/false, multiple-choice, and sentence-completion tasks, we replace the math-specific wording with a question-agnostic formulation. In particular, ``mathematical reasoning problem'' is replaced with ``question,'' and the answer-format instruction becomes
\texttt{Respond strictly using the answer format specified in the question.\textbackslash n Answer: [your answer]}.
The receiver-side scoring procedure and descending-score neighbor presentation remain unchanged.

\subsection{Byzantine Agent Behavior}
\label{app:prompts:byzantine}

Adversarial Byzantine agents (Section~\ref{sec:experiments}) bypass both prompting procedures. At each round $t$, they (i) broadcast a cached weak-model response generated from an out-of-distribution query, and (ii) report a falsified self-confidence score of $1.0$ when interacting with confidence-based baselines such as CP-WBFT~\citep{zheng2026rethinking}. This simulates a strong dishonest-confidence attack while remaining indistinguishable from honest agents at the message level. In contrast, faulty Byzantine agents follow the same prompting pipeline as reliable agents, and produce unreliable outputs solely due to the limited capability of the underlying weak model.

\section{Experimental Setup}

\subsection{Dataset for Closed-source LLMs}
\label{app:dataset}

We provide the sampled problems used in the closed-source GPT API experiments reported in Tables~\ref{tab:math500_main_sample} and~\ref{tab:math500_sample_rounds}. For Hendrycks MATH, we show 30 representative examples from the 100 sampled Level~4 instances, while for Commonsense170k, we list all 30 sampled instances. Each row reports the dataset-specific question identifier (\texttt{qid}), an abbreviated version of the question, and the ground-truth answer $y^\star$. Full problem statements can be retrieved using the corresponding \texttt{qid} from the original public releases.

\subsubsection{Hendrycks MATH (Level~4)}
\label{app:dataset:math}

The 100 sampled instances span the seven question categories present in Hendrycks MATH Level~4. We group the 30 representative examples shown below by category for readability. These problems correspond to the Hendrycks MATH experiments in Tables~\ref{tab:math500_main_sample} and~\ref{tab:math500_sample_rounds}, where \texttt{gpt-4o-mini} and \texttt{gpt-3.5-turbo} are used as the strong and weak models, respectively.

\subsubsection{Commonsense170k}
\label{app:dataset:commonsense}

The 30 sampled Commonsense170k instances span four question formats: true/false, two-option pronoun resolution, multiple-choice reasoning, and sentence completion. We group the examples by format. For non-true/false questions, the answer choices are listed below the question statement, and the ground-truth label corresponds to the option identifier in the original dataset. These problems correspond to the Commonsense170k experiments in Tables~\ref{tab:math500_main_sample} and~\ref{tab:math500_sample_rounds}, where \texttt{gpt-5} and \texttt{gpt-4o} are used as the strong and weak models, respectively.

\begin{table*}[!t]
\centering
\footnotesize
\setlength{\tabcolsep}{5pt}
\renewcommand{\arraystretch}{1.25}
\begin{tabular}{@{}l p{0.62\linewidth} r@{}}
\toprule
\textbf{qid} & \textbf{Question} & $y^\star$ \\
\midrule
\multicolumn{3}{@{}l}{\textit{Algebra}} \\
\addlinespace[2pt]
\texttt{32448964} & Value of $K$ for which $6x+4y=7,\ Kx+8y=7$ has no solution & 12 \\
\texttt{57cbe0b7} & Solve $(\sqrt{12x}+12)(\sqrt{3x}-6) = 4(x+3)+x-34$ & 50 \\
\texttt{984a9fb8} & Express $g^4 + 12g^2 + 9 = c(g^2+p)^2 + q$; find $q$ & $-27$ \\
\texttt{b0969f24} & Piecewise $f$ with $f(-4) = -60/13$, $f(4) = 3120$; find $a+b$ & 28 \\
\texttt{d026abb8} & Integer $x$ in arithmetic sequence $3^2,\, x,\, 3^4$ & 45 \\
\texttt{9bb18dfb} & Minimum of $|x-1| + |x-1.5| + |x-2|$ over $x \in \mathbb{R}$ & 1 \\
\midrule
\multicolumn{3}{@{}l}{\textit{Intermediate Algebra}} \\
\addlinespace[2pt]
\texttt{40c817ff} & $\log_y x + \log_x y = 7$; find $(\log_y x)^2 + (\log_x y)^2$ & 47 \\
\texttt{b3f26b98} & Piecewise $f$ invertible; find $k$ & 1 \\
\texttt{b11209a5} & Geometric sequence with $a_5 - a_4 = 576$, $a_2 - a_1 = 9$; sum $\sum_{i=1}^{5} a_i$ & 1023 \\
\texttt{10aef5c7} & Count of $n < 1000$ such that $\lfloor \log_2 n \rfloor$ is a positive even integer & 340 \\
\midrule
\multicolumn{3}{@{}l}{\textit{Prealgebra}} \\
\addlinespace[2pt]
\texttt{2de720d4} & Number of times the digit 6 appears in integers from 1 to 100 & 20 \\
\texttt{b6a36467} & Smallest integer $> 2$ that leaves remainder 2 modulo 3, 4, 5, 6 & 62 \\
\texttt{34e64136} & Five tests scored 87, 85, 87; last two differ by 3 with average 90; find max & 97 \\
\texttt{f3b85d7a} & Number of even perfect cubes less than 2008 & 6 \\
\texttt{c740506c} & Box nesting: 4 large $\times$ 3 medium $\times$ 2 small; total number of boxes & 40 \\
\texttt{e067503f} & Unit conversion: 6 wallops $=$ 5 ballops, 3 ballops $=$ 11 fallops; wallops for 110 fallops & 36 \\
\texttt{7bfcd56a} & First odd year after 2006 whose digits split into 3-digit and 1-digit groups with common factor $>1$ & 2013 \\
\midrule
\multicolumn{3}{@{}l}{\textit{Number Theory}} \\
\addlinespace[2pt]
\texttt{a10973bf} & Smallest integer with exactly 16 divisors, including 12 and 15 & 120 \\
\texttt{27b01b01} & Base-7 cryptarithm $\overline{AB}_7 + \overline{BA}_7 = \overline{AA0}_7$; product $A \cdot B$ & 6 \\
\texttt{583c9eaf} & Number of even positive divisors of 252 & 12 \\
\texttt{37bab629} & Count of $n \in [1, 29]$ for which $n/30$ has a repeating decimal expansion & 20 \\
\midrule
\multicolumn{3}{@{}l}{\textit{Counting \& Probability}} \\
\addlinespace[2pt]
\texttt{0a3e457d} & 20-member club, 3 distinct officers, with constraint Alex serves only if Bob does not & 6732 \\
\texttt{d790474f} & Jar with 4 red, 2 white marbles; swap-and-sample procedure; $P(\text{red}) = 11/18$ & $11/18$ \\
\texttt{b98d41f7} & Coefficient of $x^2 y^2$ in $(x+y)^4 + (x+2y)^4$ & 30 \\
\texttt{03694fa9} & 120 triangles formed by $n$ vertices on a base; find $n$ & 16 \\
\texttt{33b5135e} & 10-member chess club, 900 games played; find $N$ games per pair & 20 \\
\texttt{e5787bf4} & Distinct bracelets with 5 distinct beads under rotation and reflection equivalence & 12 \\
\midrule
\multicolumn{3}{@{}l}{\textit{Geometry}} \\
\addlinespace[2pt]
\texttt{1efe044e} & Maximum volume (cm$^3$) when right $\triangle$ with legs $3,\,4$ rotated about a leg & 50 \\
\texttt{de4ec0fd} & Isosceles $\triangle$ with $AB=AC=14$, $BC=26$; shortest angle bisector & $8\sqrt{33}/3$ \\
\midrule
\multicolumn{3}{@{}l}{\textit{Precalculus}} \\
\addlinespace[2pt]
\texttt{16be6141} & $\mathbf{v} = \mathrm{proj}_{\mathbf{a}} \mathbf{v} + \mathrm{proj}_{\mathbf{b}} \mathbf{v}$ for all $\mathbf{v}$; find $\mathbf{a} \cdot \mathbf{b}$ & 0 \\
\bottomrule
\end{tabular}
\caption{Sampled Hendrycks MATH Level~4 problems used in our experiments, grouped by category. The table shows 30 representative examples from the 100 sampled instances used throughout the evaluation. These problems are additionally used in the ablation experiments reported in Tables~\ref{tab:bw-ablation} and~\ref{tab:n9-main}. The question column gives an abbreviated statement; full problem text is recoverable from the public Hendrycks MATH release via the \texttt{qid}.}
\label{tab:math500_problems}
\end{table*}

\begin{table*}[h]
\centering
\scriptsize
\renewcommand{\arraystretch}{1}
\begin{tabular}{@{}l p{0.8\linewidth} l@{}}
\toprule
\textbf{qid} & \textbf{Question (and choices, where applicable)} & $y^\star$ \\
\midrule
\multicolumn{3}{@{}l}{\textit{True/False}} \\
\addlinespace[2pt]
\texttt{50816f9e} & Is there a season 3 of \emph{Wrecked}? & true \\
\texttt{4ba7682f} & Does Las Vegas have a professional football team? & true \\
\texttt{17f50271} & Can a person be jailed for civil contempt of court? & true \\
\texttt{4d7cce65} & Has Maroon 5 ever performed at a Super Bowl? & true \\
\texttt{84e37998} & Is \emph{I Know Why the Caged Bird Sings} a memoir? & true \\
\texttt{4bb2e906} & Is \emph{Varsity Blues} based on a true story? & false \\
\texttt{433f24bc} & Does it count if you hit the backboard on a free throw? & true \\
\texttt{b4c59eee} & Can you remove the venom glands from a snake? & true \\
\texttt{0c7abae6} & Is the Statue of Liberty in New Jersey? & false \\
\texttt{a461f8a8} & Can you send a letter without a return address? & true \\
\texttt{f448271b} & Is it illegal to pass on a solid yellow line? & false \\
\texttt{28225993} & Can you score an own goal from a direct free kick? & true \\
\texttt{90b6886c} & Do red, yellow, and orange peppers taste different? & true \\
\texttt{a03e3147} & Is there a fifth season of \emph{Mom}? & true \\
\texttt{07250c0c} & Are \emph{The Five Heartbeats} based on a real group? & false \\
\midrule
\multicolumn{3}{@{}l}{\textit{Two-option pronoun resolution}} \\
\addlinespace[2pt]
\texttt{263ebb80} & Embroidery: when threading the needle, the \_ was too thick.\newline\emph{Choices:} (option1) needle; (option2) floss. & option2 \\
\texttt{c3a49074} & ``Felicia asked Katrina about new technology because \_ was interested.''\newline\emph{Choices:} (option1) Felicia; (option2) Katrina. & option1 \\
\texttt{60091082} & Movers wanted to store the boxes in the offices, but the \_ were too small.\newline\emph{Choices:} (option1) offices; (option2) boxes. & option1 \\
\texttt{2e8510fe} & Elena loves to read books but Jessica does not; \_ bought videos all the time.\newline\emph{Choices:} (option1) Elena; (option2) Jessica. & option2 \\
\texttt{08a09099} & The public elected Jason over Randy, because \_ delivered a less persuasive speech.\newline\emph{Choices:} (option1) Jason; (option2) Randy. & option2 \\
\midrule
\multicolumn{3}{@{}l}{\textit{Multiple choice}} \\
\addlinespace[2pt]
\texttt{11aa364a} & Addison spent a month of lunches and finally found Carson a place; how would you describe Addison?\newline\emph{Choices:} (answer1) helpful; (answer2) meanspirited; (answer3) selfish. & answer1 \\
\texttt{bd59d9db} & Austin got a PS4 Pro and bought a new TV; what will Austin want to do next?\newline\emph{Choices:} (answer1) play his new game console; (answer2) return the TV; (answer3) go watch a movie at the theatre. & answer1 \\
\texttt{2ce5f4f0} & From which part of the plant does a bee get food?\newline\emph{Choices:} (answer1) flower; (answer2) seed; (answer3) stem; (answer4) root. & answer1 \\
\texttt{612a4066} & Remy came up behind Jan and pushed her in the back very hard; how would you describe Remy?\newline\emph{Choices:} (answer1) physical; (answer2) ready to fight Jan; (answer3) very hostile towards Jan. & answer3 \\
\texttt{0cd100f4} & The best way to improve future production yields on the farm is\dots\newline\emph{Choices:} (answer1) planting cabbage one year and spinach the next; (answer2) chemical fertilizers and salts; (answer3) rotating water schedules daily; (answer4) over-watering each field. & answer1 \\
\texttt{25040642} & How do you treat period cramps?\newline\emph{Choices:} (solution1) drink some coffee; (solution2) take some Midol. & solution2 \\
\texttt{ee3abe29} & Paper towel: what is it good for?\newline\emph{Choices:} (solution1) clean telescope; (solution2) operate telescope. & solution1 \\
\midrule
\multicolumn{3}{@{}l}{\textit{Sentence completion}} \\
\addlinespace[2pt]
\texttt{3abb847a} & ``How to start a giving fund'' --- after opening a bank account, what naturally follows?\newline\emph{Choices:} (ending1) choose checking vs.\ savings depending on payment frequency; (ending2) a separate account protects you from debt; (ending3) set up a line for the bank's payout the day before your event; (ending4) it isn't complicated --- consult an individual bank. & ending1 \\
\texttt{873e845e} & Javelin throw scene: as the man releases the javelin, in the background\dots\newline\emph{Choices:} (ending1) bleachers in the background, hot and sunny day; (ending2) large crowd standing around, kites flying; (ending3) a big splash, the man runs to rescue it; (ending4) another man in a black t-shirt easily catches it. & ending1 \\
\texttt{e47d4d28} & Drumming and piano duet: a child drums while a woman plays piano along; they\dots\newline\emph{Choices:} (ending1) continue playing the drums and the music; (ending2) have a small audience watching them perform; (ending3) play till there is no longer a fist drumming in the background; (ending4) play and sing along intently for joy. & ending2 \\
\bottomrule
\end{tabular}
\caption{Sampled Commonsense170k problems used in our experiments, grouped by format. These problems correspond to the closed-source commonsense reasoning experiments reported in Tables~\ref{tab:math500_main_sample} and~\ref{tab:math500_sample_rounds}, where \texttt{gpt-5} and \texttt{gpt-4o} are used as the strong and weak models, respectively. Each instance preserves its native answer format, and ground-truth answers are reported using the corresponding label as released. For multi-option items, we list the choices in the question column.}
\label{tab:commonsense_problems}
\end{table*}

\subsection{Datasets for Open-weight LLMs}
\label{app:dataset:openweight}

For the open-weight LLM experiments, we evaluate on the full MATH-500 benchmark as well as the full evaluation sets of five commonsense reasoning benchmarks: ARC-Challenge, HellaSwag, BoolQ, OpenBookQA, and RTE. These benchmarks cover a range of reasoning settings, including mathematical reasoning, commonsense inference, scientific reasoning, natural language understanding, and multiple-choice question answering. All experiments use \texttt{Qwen3-4B} as the strong model and \texttt{Qwen2.5-1.5B-Instruct} as the weak model.

\section{Additional Experimental Results}
\label{app-addexp}

In this section, we report benchmark-level and per-round results for the open-weight LLM experiments summarized in Section~\ref{openweight}. Tables~\ref{tab:appexp1} and~\ref{tab:appexp2} expand the averaged results reported in Tables~\ref{tab:math500_main} and~\ref{tab:per_round_math500_rounds}.

\paragraph{Per-benchmark robustness.}
Across all five commonsense reasoning benchmarks and topologies, SAC achieves positive BFTI on ARC-Challenge, HellaSwag, and OpenBookQA, with improvements of up to $+7.5\%$ on ARC-Challenge. On BoolQ, SAC remains within $0.6\%$ of its initial accuracy across all topologies, while on RTE it limits degradation to between $-3.4\%$ and $-0.1\%$. By contrast, CP-WBFT exhibits substantial negative BFTI on ARC-Challenge, HellaSwag, OpenBookQA, and RTE, while remaining near zero on BoolQ. The largest gap appears on RTE under the Complete topology, where CP-WBFT yields $-24.2\%$ BFTI while SAC yields only $-0.1\%$. Overall, these results suggest that SAC's receiver-side evaluation generalizes robustly across benchmarks with substantially different reasoning formats and difficulty levels.

\paragraph{Per-group behavior.}
The per-group accuracies reported in Table~\ref{tab:appexp1} show that SAC preserves or improves strong-agent accuracy on ARC-Challenge, HellaSwag, BoolQ, and OpenBookQA while limiting degradation on RTE. Representative examples include $51\to64\%$ on ARC-Challenge with Complete, $30\to34\%$ on HellaSwag with Complete, and $42\to44\%$ on OpenBookQA with Complete. Under CP-WBFT, strong-agent accuracy degrades substantially on ARC-Challenge, HellaSwag, OpenBookQA, and RTE. Weak-agent accuracy under SAC generally changes more modestly, although some degradation remains on BoolQ and RTE. H-Majority is higher under SAC across all benchmark-topology pairs, with particularly large improvements on ARC-Challenge, OpenBookQA, and RTE.

\paragraph{Per-round dynamics.}
Table~\ref{tab:appexp2} reports per-round weak/strong accuracies. Under SAC, strong-agent accuracy improves quickly and remains relatively stable on ARC-Challenge and HellaSwag. For example, strong-agent accuracy increases from $52.1\%$ at initialization to $63.7\%$ at round~6 on ARC-Challenge with MERG, and from $29.7\%$ to $33.6\%$ on HellaSwag with Complete. BoolQ and OpenBookQA show smaller changes across rounds, while RTE exhibits larger round-to-round fluctuations, with strong-agent accuracy recovering to $75.0$--$79.4\%$ by round~6 across the three topologies. In contrast, the aggregated CP-WBFT results over rounds~1--6 show substantial strong-agent degradation on ARC-Challenge, HellaSwag, OpenBookQA, and RTE.

\paragraph{Effect of topology.}
SAC exhibits relatively small performance variation across topologies within the same benchmark. FAA differs by only a few percentage points across all five commonsense reasoning benchmarks. In contrast, CP-WBFT shows larger instability depending on the graph structure, particularly on ARC-Challenge, OpenBookQA, and RTE, where FAA varies substantially across topologies while remaining consistently lower than SAC. This relatively small across-topology variance of SAC is consistent with our theoretical analysis, where the qualitative robustness of SAC primarily depends on satisfying the $(F{+}1)$-robustness condition rather than the specific topology construction.

\begin{table*}[t]
\centering
\footnotesize
\setlength{\tabcolsep}{4pt}
\begin{tabular}{lll rrrr cc r}
\toprule
Dataset & Method & Topology & IAA & FAA & BFTI & RA & W IAA$\to$FAA & S IAA$\to$FAA & H-Maj. \\
\midrule
\multirow{6}{*}{ARC-Challenge} & \multirow{3}{*}{CP-WBFT} & MERG & 43.9 & 34.1 & $-$9.8 & 32.8 & 49$\to$42 & 51$\to$37 & 32.8 \\
 &  & Complete & 44.2 & 28.4 & $-$15.8 & 32.4 & 49$\to$31 & 51$\to$32 & 32.4 \\
 &  & Erd\H{o}s--R\'enyi & 44.5 & 28.8 & $-$15.7 & 32.8 & 49$\to$33 & 52$\to$32 & 32.8 \\
\cmidrule(lr){2-10}
 & \multirow{3}{*}{SAC (Ours)} & MERG & 45.0 & 51.8 & $+$6.8 & 63.9 & 50$\to$51 & 52$\to$64 & 63.5 \\
 &  & Complete & 44.2 & 51.7 & $+$7.5 & 63.5 & 50$\to$51 & 51$\to$64 & 63.5 \\
 &  & Erd\H{o}s--R\'enyi & 44.1 & 50.6 & $+$6.5 & 62.2 & 49$\to$50 & 51$\to$62 & 62.2 \\
\midrule
\multirow{6}{*}{HellaSwag} & \multirow{3}{*}{CP-WBFT} & MERG & 28.9 & 24.4 & $-$4.4 & 23.0 & 30$\to$26 & 30$\to$24 & 23.0 \\
 &  & Complete & 28.6 & 21.3 & $-$7.3 & 22.0 & 31$\to$22 & 30$\to$22 & 22.0 \\
 &  & Erd\H{o}s--R\'enyi & 28.5 & 21.6 & $-$6.9 & 21.9 & 30$\to$22 & 30$\to$22 & 21.9 \\
\cmidrule(lr){2-10}
 & \multirow{3}{*}{SAC (Ours)} & MERG & 28.5 & 30.6 & $+$2.1 & 33.0 & 30$\to$30 & 29$\to$33 & 33.0 \\
 &  & Complete & 28.6 & 30.8 & $+$2.3 & 33.5 & 30$\to$31 & 30$\to$34 & 33.5 \\
 &  & Erd\H{o}s--R\'enyi & 28.9 & 30.6 & $+$1.8 & 33.2 & 30$\to$30 & 30$\to$33 & 33.2 \\
\midrule
\multirow{6}{*}{BoolQ} & \multirow{3}{*}{CP-WBFT} & MERG & 56.1 & 56.5 & $+$0.5 & 59.5 & 57$\to$60 & 60$\to$60 & 59.5 \\
 &  & Complete & 56.0 & 55.6 & $-$0.4 & 59.8 & 57$\to$57 & 60$\to$60 & 59.8 \\
 &  & Erd\H{o}s--R\'enyi & 56.0 & 55.8 & $-$0.1 & 60.0 & 58$\to$58 & 60$\to$59 & 60.0 \\
\cmidrule(lr){2-10}
 & \multirow{3}{*}{SAC (Ours)} & MERG & 55.9 & 55.6 & $-$0.3 & 60.3 & 58$\to$54 & 60$\to$61 & 60.1 \\
 &  & Complete & 56.4 & 55.8 & $-$0.6 & 60.4 & 57$\to$55 & 61$\to$61 & 60.6 \\
 &  & Erd\H{o}s--R\'enyi & 56.1 & 55.7 & $-$0.4 & 61.2 & 58$\to$55 & 60$\to$61 & 61.1 \\
\midrule
\multirow{6}{*}{OpenBookQA} & \multirow{3}{*}{CP-WBFT} & MERG & 36.3 & 29.3 & $-$7.0 & 26.6 & 38$\to$35 & 41$\to$31 & 26.6 \\
 &  & Complete & 36.8 & 24.3 & $-$12.5 & 26.4 & 38$\to$26 & 42$\to$26 & 26.4 \\
 &  & Erd\H{o}s--R\'enyi & 35.6 & 22.4 & $-$13.2 & 24.6 & 39$\to$24 & 40$\to$24 & 24.6 \\
\cmidrule(lr){2-10}
 & \multirow{3}{*}{SAC (Ours)} & MERG & 36.2 & 36.9 & $+$0.7 & 44.6 & 38$\to$37 & 42$\to$44 & 44.4 \\
 &  & Complete & 36.8 & 37.7 & $+$0.9 & 43.6 & 38$\to$37 & 42$\to$44 & 43.6 \\
 &  & Erd\H{o}s--R\'enyi & 36.0 & 36.1 & $+$0.1 & 42.4 & 38$\to$36 & 41$\to$42 & 42.8 \\
\midrule
\multirow{6}{*}{RTE} & \multirow{3}{*}{CP-WBFT} & MERG & 67.9 & 53.7 & $-$14.2 & 52.0 & 74$\to$65 & 79$\to$58 & 52.0 \\
 &  & Complete & 68.1 & 43.8 & $-$24.2 & 50.5 & 74$\to$47 & 79$\to$51 & 50.5 \\
 &  & Erd\H{o}s--R\'enyi & 68.2 & 45.2 & $-$23.0 & 52.0 & 75$\to$50 & 79$\to$51 & 52.0 \\
\cmidrule(lr){2-10}
 & \multirow{3}{*}{SAC (Ours)} & MERG & 68.3 & 64.9 & $-$3.4 & 75.1 & 74$\to$71 & 80$\to$75 & 76.2 \\
 &  & Complete & 68.5 & 68.4 & $-$0.1 & 79.8 & 75$\to$75 & 79$\to$79 & 80.1 \\
 &  & Erd\H{o}s--R\'enyi & 68.1 & 65.8 & $-$2.3 & 77.3 & 74$\to$72 & 79$\to$76 & 78.3 \\
\bottomrule
\end{tabular}
\caption{Full benchmark-level results corresponding to the averaged open-weight LLM results reported in Table~\ref{tab:math500_main}. We report per-dataset comparisons of CP-WBFT and SAC across three $(F{+}1)$-robust communication topologies on the full evaluation sets of five commonsense reasoning benchmarks. Each network contains $n{=}7$ agents (4 strong honest, 2 weak honest, 1 adversarial Byzantine, $F{=}3$, $r{=}4$). \texttt{Qwen3-4B} and \texttt{Qwen2.5-1.5B-Instruct} are used as the strong and weak models, respectively. W~IAA$\to$FAA and S~IAA$\to$FAA denote per-agent accuracy of the weak and strong honest groups before and after consensus; H-Majority is the fraction of queries for which the majority answer among the honest agents matches the ground truth at the final round.}
\label{tab:appexp1}
\end{table*}

\begin{table*}[t]
\centering
\footnotesize
\setlength{\tabcolsep}{4pt}
\begin{tabular}{lll ccc}
\toprule
Dataset & Method & Round & MERG (W/S) & Complete (W/S) & Erd\H{o}s--R\'enyi (W/S) \\
\midrule
\multirow{9}{*}{ARC-C} & \multirow{2}{*}{CP-WBFT} & Init & 49.0 / 51.2 & 49.2 / 51.2 & 49.5 / 51.5 \\
 &  & Rnd 1--6 & 42.0 / 37.5 & 31.4 / 32.4 & 32.6 / 32.4 \\
\cmidrule(lr){2-6}
 & \multirow{7}{*}{SAC (Ours)} & Init & 50.0 / 52.1 & 49.7 / 51.4 & 48.7 / 51.2 \\
 &  & Rnd 1 & 50.2 / 62.7 & 49.8 / 63.0 & 49.3 / 61.5 \\
 &  & Rnd 2 & 50.3 / 63.0 & 50.3 / 63.5 & 50.0 / 61.3 \\
 &  & Rnd 3 & 50.3 / 63.2 & 50.5 / 63.5 & 49.8 / 61.6 \\
 &  & Rnd 4 & 50.5 / 63.3 & 50.7 / 64.0 & 50.0 / 62.0 \\
 &  & Rnd 5 & 50.7 / 63.8 & 50.5 / 63.8 & 50.0 / 62.2 \\
 &  & Rnd 6 & 50.5 / 63.7 & 50.5 / 64.1 & 49.8 / 62.0 \\
\midrule
\multirow{9}{*}{HellaSwag} & \multirow{2}{*}{CP-WBFT} & Init & 30.2 / 30.1 & 30.6 / 30.3 & 30.4 / 29.9 \\
 &  & Rnd 1--6 & 25.9 / 24.5 & 21.6 / 22.0 & 22.1 / 21.9 \\
\cmidrule(lr){2-6}
 & \multirow{7}{*}{SAC (Ours)} & Init & 30.4 / 29.3 & 30.4 / 29.7 & 30.4 / 29.9 \\
 &  & Rnd 1 & 29.7 / 32.1 & 30.8 / 32.9 & 29.6 / 32.7 \\
 &  & Rnd 2 & 29.8 / 33.0 & 30.6 / 33.4 & 30.0 / 33.2 \\
 &  & Rnd 3 & 29.9 / 33.2 & 30.8 / 33.5 & 30.1 / 33.4 \\
 &  & Rnd 4 & 29.9 / 33.2 & 30.6 / 33.6 & 29.9 / 33.2 \\
 &  & Rnd 5 & 29.9 / 33.1 & 30.9 / 33.6 & 29.9 / 33.2 \\
 &  & Rnd 6 & 29.9 / 33.1 & 30.6 / 33.6 & 29.9 / 33.2 \\
\midrule
\multirow{9}{*}{BoolQ} & \multirow{2}{*}{CP-WBFT} & Init & 57.2 / 60.4 & 57.3 / 60.3 & 57.6 / 60.0 \\
 &  & Rnd 1--6 & 60.0 / 59.9 & 56.8 / 59.8 & 58.5 / 59.3 \\
\cmidrule(lr){2-6}
 & \multirow{7}{*}{SAC (Ours)} & Init & 57.5 / 60.0 & 57.5 / 60.9 & 57.5 / 60.4 \\
 &  & Rnd 1 & 52.5 / 61.3 & 52.0 / 61.5 & 52.7 / 62.6 \\
 &  & Rnd 2 & 54.9 / 61.0 & 54.3 / 61.1 & 55.0 / 61.0 \\
 &  & Rnd 3 & 52.6 / 61.5 & 52.0 / 61.1 & 53.1 / 62.8 \\
 &  & Rnd 4 & 54.3 / 61.3 & 54.6 / 60.9 & 55.1 / 61.3 \\
 &  & Rnd 5 & 52.6 / 62.0 & 52.1 / 60.9 & 53.0 / 62.1 \\
 &  & Rnd 6 & 54.3 / 61.0 & 54.5 / 61.2 & 55.0 / 61.1 \\
\midrule
\multirow{9}{*}{OBQA} & \multirow{2}{*}{CP-WBFT} & Init & 38.1 / 41.0 & 38.4 / 42.1 & 38.8 / 40.4 \\
 &  & Rnd 1--6 & 34.6 / 30.6 & 26.2 / 26.4 & 24.5 / 24.4 \\
\cmidrule(lr){2-6}
 & \multirow{7}{*}{SAC (Ours)} & Init & 38.1 / 41.9 & 38.0 / 41.9 & 38.2 / 41.4 \\
 &  & Rnd 1 & 37.2 / 44.0 & 37.5 / 44.9 & 36.8 / 42.1 \\
 &  & Rnd 2 & 37.5 / 43.6 & 37.3 / 44.0 & 36.6 / 42.0 \\
 &  & Rnd 3 & 37.2 / 43.6 & 37.5 / 44.1 & 36.6 / 42.3 \\
 &  & Rnd 4 & 37.3 / 43.9 & 37.3 / 43.6 & 36.6 / 42.5 \\
 &  & Rnd 5 & 37.0 / 43.6 & 37.3 / 43.9 & 36.6 / 42.5 \\
 &  & Rnd 6 & 37.2 / 43.5 & 37.3 / 43.9 & 36.5 / 42.4 \\
\midrule
\multirow{9}{*}{RTE} & \multirow{2}{*}{CP-WBFT} & Init & 74.0 / 78.9 & 74.4 / 79.1 & 74.7 / 79.1 \\
 &  & Rnd 1--6 & 65.3 / 58.4 & 46.6 / 50.5 & 50.0 / 51.2 \\
\cmidrule(lr){2-6}
 & \multirow{7}{*}{SAC (Ours)} & Init & 74.0 / 79.6 & 75.3 / 79.4 & 73.6 / 79.4 \\
 &  & Rnd 1 & 70.4 / 65.3 & 74.9 / 69.9 & 70.6 / 66.9 \\
 &  & Rnd 2 & 71.7 / 77.2 & 74.9 / 77.4 & 72.2 / 76.9 \\
 &  & Rnd 3 & 70.4 / 68.4 & 74.5 / 69.9 & 70.6 / 67.3 \\
 &  & Rnd 4 & 71.5 / 75.8 & 74.9 / 78.1 & 72.2 / 75.9 \\
 &  & Rnd 5 & 70.6 / 67.2 & 74.5 / 68.7 & 70.6 / 67.3 \\
 &  & Rnd 6 & 71.3 / 75.0 & 74.9 / 79.4 & 72.0 / 76.2 \\
\bottomrule
\end{tabular}
\caption{Benchmark-level per-round results corresponding to the averaged open-weight LLM results reported in Table~\ref{tab:per_round_math500_rounds}. We report per-round weak / strong honest accuracy (W/S, in \%) across three $(F{+}1)$-robust communication topologies on the full evaluation sets of five commonsense reasoning benchmarks. Each cell shows the average accuracy of the weak honest group (\texttt{Qwen2.5-1.5B-Instruct}, 2 agents) and the strong honest group (\texttt{Qwen3-4B}, 4 agents) at the corresponding round under $n{=}7$ agents and $F{=}3$.}
\label{tab:appexp2}
\end{table*}

\clearpage
\twocolumn[
\section{Licenses and Terms of Use}
]
\label{app:licenses}
All datasets and models used in this work are publicly available and used in compliance with their respective licenses and terms of use. The Hendrycks MATH dataset~\citep{MATHall} is released under the MIT License. Commonsense170k~\citep{hu2023llm} is released under the Apache 2.0 License. ARC-C~\citep{clark2018arc} is released under the CC BY-SA 4.0 License. HellaSwag~\citep{zellers2019hellaswag} is released under the MIT License. BoolQ~\citep{clark2019boolq} is released under the CC BY-SA 3.0 License. OBQA~\citep{mihaylov2018openbookqa} is released under the Apache 2.0 License. RTE~\citep{rte} is distributed for research use as part of the PASCAL RTE challenge. The Qwen2.5-1.5B-Instruct and Qwen3-4B models are released under the Apache 2.0 License. The GPT models (gpt-3.5-turbo, gpt-4o-mini, gpt-4o, gpt-5) are accessed via the OpenAI API in compliance with OpenAI's Terms of Use. All artifacts are used solely for non-commercial research purposes, consistent with their intended use.

\end{document}